\documentclass[11pt]{article}

\usepackage{graphicx}
\usepackage{xcolor}
\usepackage{hyperref}

\usepackage{amsmath,amssymb,epsf,epic}
\usepackage{tikz,tkz-tab}

\numberwithin{equation}{section}

\newcommand{\BRA}{\big\langle\hspace{-.1cm}\big\langle} 
\newcommand{\KET}{\big\rangle\hspace{-.1cm}\big\rangle}

\newcommand{\un}{{\mathbb I}}
\newcommand{\ra}{\rightarrow}
\newcommand{\tr}{{\,\rm Tr}}
\newcommand{\ran}{{\,\rm Ran}}
\newcommand{\spa}{{\,\rm Span}}
\renewcommand{\ker}{{\, \rm Ker\,}}

\newcommand{\bra}{\langle}

\newcommand{\ket}{\rangle}

\renewcommand{\i}{{\rm i}}

\newcommand{\be}{\begin{equation}}
\newcommand{\ee}{\end{equation}}
\newcommand{\bea}{\begin{eqnarray}}
\newcommand{\eea}{\end{eqnarray}}

\newcommand{\eps}{\varepsilon}

\newcommand{\ffi}{\varphi}

\newcommand{\ode}{{\cal O}}
\newcommand{\e}{{\rm e}}
\newcommand{\grintl}{[\kern-.18em [}
\newcommand{\grintr}{]\kern-.18em ]}

\newcounter{resultcounter}[section]

\newtheorem{thm}[resultcounter]{Theorem}
\newtheorem{lem}[resultcounter]{Lemma}
\newtheorem{prop}[resultcounter]{Proposition}
\newtheorem{cor}[resultcounter]{Corollary}

\newtheorem{rem}[resultcounter]{Remark}

\def\cA{{\cal A}} \def\cB{{\cal B}} 
\def\cD{{\cal D}}  
\def\cG{{\cal G}} \def\cH{{\cal H}} 
 \def\cK{{\cal K}} \def\cL{{\cal L}}
  
\def\cP{{\cal P}} \def\cQ{{\cal Q}} \def\cR{{\cal R}}
\def\cS{{\cal S}} \def\cT{{\cal T}} \def\cU{{\cal U}}
\def\cV{{\cal V}} \def\cW{{\cal W}} \def\cX{{\cal X}}
\def\cY{{\cal Y}} \def\cZ{{\cal Z}}

\newcommand{\R}{{\mathbb R}}
\newcommand{\N}{{\mathbb N}}

\newcommand{\C}{{\mathbb C}}

\renewcommand{\P}{{\mathbb P}}

\def\proof{\noindent{\bf Proof:}\ \ }
\def\qed{\hfill $\Box$\medskip}

\begin{document}
\title{Adiabatic  Lindbladian Evolution with Small Dissipators}
 \author{ Alain Joye\footnote{ Univ. Grenoble Alpes, CNRS, Institut Fourier, F-38000 Grenoble, France} }

\date{ }

\maketitle

\maketitle
\vspace{-1cm}

\thispagestyle{empty}
\setcounter{page}{1}
\setcounter{section}{1}

\setcounter{section}{0}

\vskip 1cm

\noindent{\bf Abstract}: { We consider a time-dependent small quantum system weakly coupled to an environment, whose effective dynamics we address by means of a Lindblad equation. We assume the Hamiltonian part of the Lindbladian is slowly varying in time and  the dissipator part has small amplitude. We study the properties of the evolved state of the small system as the adiabatic parameter and coupling constant both go to zero, in various asymptotic regimes. In particular, we analyse the deviations of the transition probabilities of the small system between the instantaneous eigenspaces of the Hamiltonian with respect to their values in the purely Hamiltonian adiabatic setup, as a function of both parameters. }

\medskip 


\medskip 

\noindent {\bf Acknowledgments}: This work is partially supported by the ANR grant NONSTOPS (ANR-17-CE40-0006-01)

\section{Introduction}

The adiabatic approximation of quantum mechanics, designed to address the time-dependent Schr\"odinger equation in its Hamiltonian formulation, has been introduced very soon after the discovery of quantum mechanics \cite{BF, La, Z}. It has since been developed in order to accommodate more general Hamiltonians and to improve its accuracy; see \cite{ K1, N1, ASY, T, AE, N2, JKP, JP0, J1, JP, Sch, BDF} for examples along these lines. The adiabatic approximation being instrumental in the analysis of time-dependent phenomena, this mathematical method was extended and applied to a variety of evolution equations in more general contexts, like discrete time evolutions, \cite{DKS, HJPR1, HJPR2},  non-linear setups \cite{CFK1, CFK2, GG, S, LLFY,  F-KJ}, or contracting evolutions in Banach spaces \cite{Kr, NR, J2, AFGG} for example.

Specifically, adiabatic approaches were successfully adapted to address the evolution of open quantum systems consisting of a time-dependent small system of interest coupled to an environment, be it from a global Hamiltonian perspective encompassing a modeling of the environment, or from an effective point of view through a Lindblad evolution equation. In particular, asymptotic expressions for quantum states solution to an evolution equation driven by time-dependent Lindblad generators in the adiabatic limit are provided in \cite{AFGG, AFGG2, FH}, together with detailed analyses of the special case of dephasing Lindbladians. See  \cite{DS, TW, JMS} for results along the same lines, including a Hamiltonian description of an environment at positive and zero temperature, while \cite{A-SF, HJPR1, HJPR2, BFJP} focus on entropy production issues in the adiabatic regime of such systems.

From the point of view of applications to quantum engineering and quantum control, these adiabatic approaches are suited to describe the evolution of a small system that can be monitored in a time dependent fashion by external agents, and which is weakly coupled to its environment, due to imperfect isolation. In the adiabatic limit, and for a regime of small coupling, one expects to get a description of the evolution of the small system in terms of the characteristics of the Hamiltonian, with quantitative information of the perturbations induced by the effect of the environment. While these questions have been addressed in the physics literature in various setups, see {\it e.g.} the discussions and references in \cite{AFGG2, FH}, the mathematical approaches of those questions are less numerous and often concern specific cases.  

Detailed information about the evolution of the small system is available for dephasing Lindbladians only \cite{AFGG, AFGG2, FH, H}, in which by definition, the dissipator is a function of the Hamiltonian. Hence, a time dependence in the Hamiltonian implies a similar time dependence in the dissipator as well.  In the Hamiltonian model addressed in \cite{JMS}, the coupling between the small system and the (bosonic) environment is assumed to be energy conserving at all times, which makes it  dependent on the Hamiltonian as well. This is arguably a shortcoming of the approach since the dissipator models the effect of the coupling of the small system to the environment that, in general, is likely to be independent of the way the system is monitored and is possibly time-independent. In the model considered in \cite{DS}, the coupling of the small system to the (fermionic) environment is time independent, however the coupling constant is determined by the adiabatic parameter.

In case the Lindblad generator is time independent, the dynamics of states can be inferred from the spectral properties of the Lindbladian. Therefore, a host of perturbative methods have been designed to identify and approximate the asymptotic state, as well as to analyse more precise properties of the dynamics, {\it e.g.} unravelling. See for example the recent papers \cite{aletal, baletal, macietal, HJ, benetal}, and the references therein for works along these lines, in various setups.

The present contribution is devoted  to the study of the effective dynamics of a small quantum system weakly coupled to an environment, assuming a Lindbladian description that consists in a slowly varying time-dependent Hamiltonian drive and an arbitrary dissipator. We analyse the evolved state of the small quantum system as both the adiabatic parameter $\eps>0$ and the coupling constant $g>0$ vanish, in an independent way. Actually we can, and will, consider time-dependent dissipators, which will allow for comparisons with some of the results mentioned above. 

Under suitable assumptions, we provide leading order approximations of the density matrix of the small system in the {\it perturbative regime}, $g\ll \eps$, where the coupling constant is much smaller than the adiabatic parameter, 
in the {\it slow drive regime}, $\eps\ll g$, where the adiabatic parameter is much smaller than the coupling constant, 
as well as in a {\it transition regime},  $g\ll \sqrt \eps$, bridging the gap between the two previous regimes. The transition regime is addressed by means of a {\it reduced dynamics} defined on the kernel of the Hamiltonian part of the Lindbladian, that depends on the single parameter $\eps/g$ that determines the regime we are in. We show that the reduced dynamics approximates the asymptotic Lindbladian evolution in the transition regime, that covers the perturbative regime and, partially, the slow drive regime.

As a consequence, we also derive the asymptotics of the transition probabilities between instantaneous eigenspaces of the Hamiltonian in these regimes. In the perturbative regime, the leading order of these transition probabilities is shown to be given by the familiar expression of order $\eps^2$ depending only on the Hamiltonian  if $g\ll \eps^3$, 
and by an explicit integral expression of order $g/\eps$ that depends on the dissipator  if $\eps^3\ll g\ll \eps$. This is in keeping with \cite{JMS} where a similar transition in the asymptotics of the transition probability was observed for the Hamiltonian model considered. In case $g=\eps$, we are in the transition regime and the reduced dynamics is independent of $\eps$ to leading order. This regime corresponds, in spirit, to the regime addressed in their Hamiltonian model by \cite{DS}, Section 3.  For the slow drive regime, $\eps\ll g$, we get that the transition probability is independent of $\eps$ and $g$ to leading order, and is characterised by the kernel of the dissipator.

These features are illustrated for a two-level system with a dissipator displaying a certain symmetry. This allows for the explicit computation of the reduced dynamics which interpolates between these regimes.

\section{Setup and main results}\label{setup}

The separable Hilbert space of the small system is denoted by $\cH$ and $t\mapsto H(t)$ is its time dependent Hamiltonian on $\cH$. The dissipator of the Lindbladian is constructed by means of  a finite sum (for simplicity) of bounded operators on $\cH$, called jump operators, $t\mapsto \Gamma_l(t)$, $l\in I$, a  set of indices. We will assume the following regularity hypotheses:
\medskip 

\noindent{\bf Reg}\\
$\bullet$ $H: [0,1]\ra \cB(\cH)$ is self-adjoint valued and $C^\infty$ in norm (with right and  left derivatives at $\{0,1\}$).\\
$\bullet$ $\partial_t^k H(t)|_{t=0}=0$, for all $k\in \N^*$.\\
$\bullet$ For each $j\in I$, $I$  a finite set of indices, $\Gamma_j: [0,1]\ra \cB(\cH)$ is  $C^\infty$ in norm (with right and left derivatives at $\{0,1\}$). \\

Note that while the leading order results stated in the present section do not require all derivatives of the Hamiltonian at zero to vanish, the arbitrary high order generalisations of Section \ref{secgen} require this property. This assumption ensures that high order adiabatic approximations of the Heisenberg unitary evolution of the spectral projectors of the Hamiltonian coincide with the spectral projectors of the Hamiltonian at time zero.
\medskip

For $g\geq 0$, and each $t\in [0,1]$, the time-dependent Lindblad operator $\cL_t^{[g]}(\cdot)\in \cB(\cB(\cH))$ reads
\be\label{lind}
\cL_t^{[g]}(\cdot)=\cL_t^0(\cdot)+g\cL^1_t(\cdot)=-\i[H(t),\cdot]+g\sum_{l\in I}\Big(\Gamma_l(t)\cdot \Gamma_l^*(t)-\frac12\big\{\Gamma_l^*(t)\Gamma_l(t), \cdot\big\}\Big),
\ee
where the Hamiltonian part $\cL_t^0(\cdot)$ is time-dependent, while the dissipator $g\cL^1_t(\cdot)$ is possibly constant. 
The Lindbladian $\cL_t^{[g]}$ acts in particular on the Banach space $\cT(\cH)$, the set of trace class operators on $\cH$ with norm denoted by $\|\cdot\|_1$.
Further specialising, the Lindbladian acts on the set of density matrices or states, {\it i.e.} positive trace class operators of trace one, in the Schr\"odinger picture we adopt here.  
\medskip

A special case of interest for which the dissipator depends on time is that of {\it dephasing Lindbladians} characterized by $\Gamma_l(t)=F_l(H(t))$, where $F_l: \R\ra \C$ is some smooth function, for each $l\in I$,  see \cite{AFGG}. Among other things, dephasing Lindbladians enjoy the following properties for each $t$ fixed: 
\begin{align}\label{deph}
&\ker \cL_t^{[g]}(\cdot)=\ker [H(t), \cdot ] \ \ \mbox{in } \ \cT(\cH), \nonumber\\ 
&[\Gamma_j(t),P(t)]=0, \ \ \forall \ \mbox{spectral projector} \ P(t) \ \mbox{of} \ H(t).
\end{align}

We shall work on $\cT(\cH)$, unless stated otherwise, and the corresponding operator norm of $\cA\in \cB(\cT(\cH))$ will be denoted by $\|\cA\|_{\tau}$.  In particular, { for any $A\in\cB(\cH)$, the maps on  $ \cT(\cH)\ni \rho$ given by $\cA_l: \rho \mapsto A \rho$ and $\cA_r: \rho \mapsto \rho A$ belong to $ \cB(\cT(\cH))$} and have norms satisfying $\|\cA_\#\|_\tau\leq \|A\|$, $\#\in\{l,r\}$, where $\|\cdot \|$ denotes the operator norm on $\cB(\cH)$. For $\cA, \cB$ two operators in $\cB(\cB(\cH))$, we will denote their composition by $\cA\circ \cB$, or simply $\cA \cB$ if no risk of confusion arises.
\medskip

For $\eps>0$, $g\geq 0$, we consider the Lindblad equation
\be\label{lindeq}
\left\{\begin{matrix}
\eps \dot \rho= (\cL_t^0+g\cL^1_t)(\rho),  & t\in [0,1],\cr  \rho(0)=\rho_0, 
\hfill &\phantom{x}\rho\in \cT(\cH),
\end{matrix} \right.
\ee
in the adiabatic and small coupling regimes, characterised by $(\eps, g)\ra (0,0)$.\\

We recall here the main properties of the solutions to (\ref{lindeq}). As is well known, see \cite{D, L} and is recalled in \cite{AFGG, H} for example, for each fixed $t\in [0,1]$, the one-parameter family $(\e^{s\cL^{[g]}_t})_{s\geq 0}$ considered  on $\cT(\cH)$ forms a norm continuous semigroup of completely positive and  trace preserving  (CPTP) applications,  which are contraction operators. 
Consequently, denoting by $(\cU(t,s))_{0\leq s\leq t\leq 1}$ the two-parameter propagator associated to (\ref{lindeq}),
\begin{align}\label{ulind}
\left\{\begin{matrix}
\eps \partial_t\cU(t,s)=(\cL_t^0+g\cL^1_t)(\cU(t,s)), \cr
\cU(s,s)=\un, \ \ 0\leq s \leq t \leq 1,\hfill 
\end{matrix} \right.
\end{align}
it follows from Thm X.70 in \cite{RS} for example that the propagator is a contraction:
\begin{align}
\|\cU(t,s)\rho \|_1 \leq \|\rho\|_1,\phantom{2} \  \ &\forall  \rho \in \cT(\cH), \ \forall 1\geq  t\geq s\geq 0.
\end{align}
In particular we have $\|\cU(t,s)\|_\tau= 1$, since $\cU(t,s)$ is trace preserving.\\
While we shall stick to the bounded case, note that unbounded Hamiltonians and/or dissipators could also be accommodated, \cite{D, J2, FFFS}, for example.\\

We suppose that the spectrum of the Hamiltonian is separated into several disjoint subsets, which corresponds to the familiar gap assumption of the adiabatic theory. 

\medskip
\noindent 
{\bf Spec}\\
For $2\leq d < \infty$, there exists $G>0$ such that for all $t\in [0,1]$, the spectrum of $H(t)$, $\sigma(H(t))$, satisfies 
\be
\sigma(H(t))=\cup_{1\leq j\leq d}\, \sigma_j(t),  \ \ \ \  \inf_{t\in[0,1], 1\leq j\neq k\leq d}{\rm dist} (\sigma_j(t),\sigma_k(t))\geq G >0.
\ee
Accordingly, we introduce the corresponding self-adjoint spectral projectors on $\cH$  for $1\leq j\leq d$
\be\label{rieszj}
P_j(t)=-\frac{1}{2\i \pi}\oint_{\gamma_j}(H(t)-z)^{-1}dz, 
\ee
where $\gamma_j\in \rho(H(t))$ is a positively oriented simple loop encircling $\sigma_j(t)$ which contains no element of $\sigma(H(t))\setminus \sigma_j(t)$ in its interior,  {\it i.e.} ${\rm int }\, \gamma_j\cap \sigma(H(t))= \sigma_j(t)$. \\
Moreover, $P_j: t\mapsto P_j(t)$ is $C^\infty$ since $H$ is, and
\begin{align}\label{specproj}
&P_j(t)P_k(t)=\delta_{jk}P_j(t),\ \ \ \sum_{1\leq j\leq d} P_j(t) = \un.
\end{align}
Note that for all $t\in[0,1]$ and $1\leq j\leq d$, $P_j(t)$ belongs to $\ker \cL_t^0$, {\it i.e.} $\cL_t^0(P_j(t))\equiv 0$. 
Following \cite{K1}, we  introduce the operator on $\cH$
\be\label{multik}
K(t)=\sum_{1\leq j\leq d} P'_j(t)P_j(t)=-\sum_{1\leq j\leq d} P_j(t)P'_j(t),
\ee
and the corresponding {\it parallel transport, or Kato,} operator on $\cH$ solution to 
\bea\label{katop}
\left\{\begin{matrix}
\partial_t W(t,s)=K(t) W(t,s),\hfill \\
W(s,s)=\mathbb I, \ \ 0\leq s,t \leq 1.
\end{matrix} \right.
\eea
It is unitary and satisfies the well known intertwining relation
\be\label{interk}
W(t,s)P_j(s)=P_j(t)W(t,s),
\ee
whose proof is based on the fact that for any smooth projector $P(t)=P^2(t)$, $P(t)P'(t)P(t)\equiv 0$, see \cite{K1, K2, Kr}.
Note that the propagator $(W(t,s))_{0\leq s,t\leq 1}$ is actually well defined and invertible for any set of projectors that satisfy (\ref{specproj}) in a Banach space framework, and it satisfies the intertwining relation (\ref{interk}) for all $1\leq j\leq d$.\\

To keep technicalities to a minimum in this presentation section, we state the main results of the paper in their leading order formulations, and under simple assumptions. As will be mentioned along the way, some results  are corollaries of more general statements to be found in later sections. {The last paragraph of the present section indicates the locations in the manuscript where the proofs of the results and their generalisations are to be found.}

\subsection{Perturbative regime $g\ll \eps\ll 1$}

Our first result describes the modification of the adiabatic transition probabilities between the spectral subspaces $P_j(0)\cH$ at time zero and $P_k(t)\cH$ at time $t$, { for $j\neq k$,}  induced by the presence of the dissipator $g\cL^1_t$ in the regime $g\ll\eps$. 

{Pick an initial state $\rho_j\in\cT(\cH)$ such that $\rho_j=P_j(0)\rho_j P_j(0)$ }
and denote by $\cU^0(t,s)$ the solution to (\ref{ulind}) with $g=0$. In absence of dissipator, the transition probability considered reads { $\tr (P_k(t)\cU^0(t,0)(\rho_j))$} and is of order $\eps^2$, see Proposition \ref{propuread} and Remark \ref{remtiltra}. In case both spectral projectors involved 
are associated to a (potentially degenerate) eigenvalue, $\sigma_j(t)=\{e_j(t)\}$, $\sigma_k(t)=\{e_k(t)\}$, one has the explicit expression (\ref{pureadev}) for $j\neq k$
\begin{align}\label{indema}
 & \tr (P_k(t)\cU^0(t,0)(\rho_j)) 
 =\eps^2\tr \Big\{\frac{P_k(t)P_k'(t)\tilde\rho_j(t)P_k'(t) P_k(t)}{(e_j(t)-e_k(t))^2}\Big\}+\ode(\eps^3), \ \mbox{where} \nonumber\\
 &\tilde\rho_j(t)=W(t,0)\rho_jW(0,t).
\end{align}
Note that due to (\ref{interk}) $\tilde \rho_j(t)=P_j(t)\tilde \rho_j(t) P_j(t)$.

When the dissipator term is turned on we have, in the perturbative regime:
\begin{thm}\label{maing<e} Assume {\bf Reg} and {\bf Spec} with  $\sigma_j(t)=\{e_j(t)\}$ for all $t\in [0,1]$, and consider a state {$\rho_j=P_j(0)\rho_j P_j(0)$}. Then, the solution to (\ref{ulind}) satisfies for $j\neq k$, as $(\eps,g)\ra (0,0)$ with $g/\eps\ra 0$, 
 \begin{align}\label{sigjej}
 \tr (P_k(t)\cU(t,0)(\rho_j))=& \tr (P_k(t)\cU^0(t,0)(\rho_j)) \nonumber\\
 &+\frac{g}{\eps}\sum_{l\in I}\int_0^t \tr (P_k(s)\Gamma_l(s)\tilde\rho_j(s)\Gamma_l^*(s)P_k(s))ds+\ode(g+g^2/\eps^{2}),
 \end{align}
 with $\tilde \rho_j(t)=W(t,0)\rho_jW(0,t)$.
 
Further assuming  that for $j\neq k$, $\sigma_j(t)=\{e_j(t)\}$ and $\sigma_k(t)=\{e_k(t)\}$  for all $t\in [0,1]$, we have in the same regime, 
 \begin{align}\label{214}
 \tr (P_k(t)\cU(t,0)(\rho_j))=&\ \eps^2\tr \Big\{\frac{P_k(t)P_k'(t)\tilde \rho_j(t)P_k'(t) P_k(t)}{(e_j(t)-e_k(t))^2}\Big\}\nonumber\\
 &+\frac{g}{\eps}\sum_{l\in I}\int_0^t \tr (P_k(s)\Gamma_l(s)\tilde\rho_j(s)\Gamma_l^*(s)P_k(s))ds+\ode(g+\eps^3+g^2/\eps^{2}),
 \end{align}
 with $\tilde \rho_j(t)=W(t,0)\rho_jW(0,t)$.
\end{thm}
The physical interpretation is that in this regime, the dissipator contributes to the adiabatic transition probabilities of order $\eps^2$ by a history dependent perturbative term of order $g/\eps$; see also Theorem \ref{asuper} below.
\begin{rem} i) The correction to the transition probability due to the dissipator is non negative.  \\
ii) If we drop the assumption $\sigma_j(t)=\{e_j(t)\}$, formula (\ref{sigjej}) still holds with $\tilde \rho_j(s)$ replaced by a state $\tilde \rho_j(s, \eps)$ that depends on $\eps$ and also satisfies { $\tilde \rho_j(s, \eps)=P_j(t)\tilde \rho_j(s, \eps)P_j(t)$}, see (\ref{rhofi}). 
In any case, if  $P_j(0)$ is of finite rank and $\rho_j=P_j(0)/\dim(P_j(0))$, then (\ref{sigjej}) holds with $\tilde \rho_j(t)=P_j(t)/\dim(P_j(0)).$ \\
iii) The transition probability from $P_j(0)\cH$ to  $P_j(t)\cH$ reads
\begin{align}
\tr (P_j(t)\cU(t,0)(\rho_j))=&\tr (P_j(t)\cU^0(t,0)(\rho_j))\nonumber\\
&-\frac{g}{\eps}\sum_{l\in I}\int_0^t \tr ((\un-P_j(s))\Gamma_l(s)\tilde\rho_j(s)\Gamma_l^*(s)(\un-P_j(s)))ds+\ode(g+g^2/\eps^{2}).
\end{align}
\\
iv) In case $g=\eps^3$, both contributions in (\ref{214}) are of order $\eps^2$ and the error term is $\ode(\eps^3)$.\\
v) If $\eps^3\ll g \ll \eps$, the dissipator contribution takes over, with arbitrary slow decay
\begin{align}
 \tr (P_k(t)\cU(t,0)(\rho_j))=\frac{g}{\eps}\sum_{l\in I}\int_0^t \tr (P_k(s)\Gamma_l(s)\tilde\rho_j(s)\Gamma_l^*(s)P_k(s))ds+\ode(g+\eps^2+g^2/\eps^{2}).
 \end{align} 
 vi) If $g\ll \eps^3$, one recovers the adiabatic result to leading order
 \begin{align}
 \tr (P_k(t)\cU(t,0)(\rho_j))=&\ \eps^2\tr \Big\{\frac{P_k(t)P_k'(t)\tilde\rho_j(t)P_k'(t) P_k(t)}{(e_j(t)-e_k(t))^2}\Big\}+\ode(\eps^3).
\end{align}
vii) If $\cL_t^{[g]}$ is dephasing, the contribution of the dissipator vanishes, due to (\ref{deph}). This is keeping with Thm 18 of \cite{AFGG}  which yields transition probabilities of order $\eps g$ with our notations; see also \cite{AFGG2}.\\
viii)  Finally, Theorem \ref{maing<e} is a consequence of Theorem \ref{asuper} stated below.

\end{rem}
 
While Theorem \ref{maing<e} focuses on transition probabilities, we also provide higher order approximations of the full propagator $\cU(t,s)$ in Propositions \ref{propweak} and \ref{propweaksuper}. The full formulations are too  involved for this presentation section and we limit ourselves here to the leading order expression stated  as Theorem \ref{asuper}.

Any state $\rho\in \cT(\cH)$ can be written as
\be
\rho=\sum_{1\leq n,m\leq d} P_n(t)\rho P_m(t),
\ee
with  off diagonal elements, or coherences, $P_n(t)\rho P_m(t)$, for $n\neq m \in \{1,\dots, d\}$, and  diagonal elements, or populations, $P_n(t)\rho P_n(t)$, for $n \in \{1,\dots, d\}$.
The extraction of the diagonal part of $\rho$ is obtained by the action of the projector $\cP_0(t)$ on $\cB(\cH)$ defined for any $A\in \cB(\cH)$ by 
\be\label{projoker}
\cP_0(t)(A)=\sum_{1\leq n\leq d} P_n(t) A P_n(t).
\ee
{ We take advantage of the fact that the superoperator $\cP_0(t)$ acts on states in $\cT(\cH)$ in the same way as its dual acts on observables in $\cB(\cH)$. When acting on $\cT(\cH)$,  $\cP_0(t)$} is a CPTP map characterised by its Kraus operators. In case $\sigma_k(t)=\{e_k(t)\}$ for all $1\leq k\leq d$ and all $t\in [0,1]$, $\cP_0(t)$ coincides with the spectral projector onto $\ker \cL_t^0$.

Then, observe that under {\bf Reg},
the projector $\cP_0(t)$ 
given by (\ref{projoker}) is smooth in trace norm.
Consider the  parallel transport operator $\cW_0(t,s)$, $0\leq t,s\leq 1$,  associated to $\cP_0(t)$ via the equation
\bea\label{w0}
\left\{\begin{matrix}
\partial_t\cW_0(t,s)=[\cP_0'(t),\cP_0(t)]\cW_0(t,s), \\ 
\cW_0(s,s)=\un \hfill \end{matrix}\right.
\eea
that satisfies, , see \cite{Kr}, the intertwining relation
\be\label{intertw0}
\cP_0(t)\cW_0(t,s)=\cW_0(t,s)\cP_0(s)
\ee
and the propagation relation for all $0\leq r,s,t\leq 1$
\be\label{proprel}
\cW_0(t,s)\cW_0(s,r)=\cW_0(t,r).
\ee
The operator $\cW_0(t,s)$ enjoys further properties: 
\begin{lem}\label{wpcptp} Assume {\bf Reg} and {\bf Spec} Then, for any $0\leq s, t\leq 1$, the  operator $\cW_0(t,s)\cP_0(s)$  is a  CPTP map on $\cT(\cH)$,  and $\cW_0(t,s)$ maps $\ran \cP_0(s)$ to 
$\ran \cP_0(t)$ isometrically in trace norm.
Moreover, 
\be\label{par0pdimstate}
\cW_0(t,s) \cP_0(s)(\rho)=W(t,s)  \cP_0(s)(\rho) W(s,t),
\ee 
where $W(t,s)$ is the Kato operator defined by (\ref{katop}). 
\end{lem}
\begin{rem}
The first statements are shown  in \cite{AFGG}, and (\ref{par0pdimstate}) is proven below in lemma \ref{lemw}.
\end{rem}
We are now ready to give the approximation of $\cU(t,0)$:

\begin{thm}\label{asuper}
Assume {\bf Reg} and {\bf Spec} with  $\sigma_k(t)=\{e_k(t)\}$ for all $1\leq k\leq d$ and all $t\in[0,1]$. Then the solution to (\ref{ulind}) satisfies, as $(\eps,g)\ra (0,0)$ with $g/\eps\ra 0$, 
\begin{align}\label{apupr}
\cU(t,0)\cP_0(0)&=\cU^0(t,0)\cP_0(0)+\frac{g}{\eps}\int_0^t \cP_0(t)\cW_0(t,s)\cP_0(s) \cL_s^1 \cW_0(s,0)\cP_0(0)ds+\ode(g+(g/\eps)^2)\\ \nonumber
&=\cP_0(t)\cW_0(t,0)\cP_0(0)+\frac{g}{\eps}\int_0^t \cP_0(t)\cW_0(t,s)\cP_0(s) \cL_s^1 \cW_0(s,0)\cP_0(0)ds+\ode(\eps+(g/\eps)^2),
\end{align}
where, for all $A\in\cB(\cH)$, 
\be
 \cW_0(t,s)\cP_0(s)(A)=\cP_0(t)\cW_0(t,s)\cP_0(s)(A)=\sum_{1\leq n\leq d}W(t,s)P_n(s)AP_n(s)W(s,t),
 \ee
with $W(t,s)$ the Kato operator (\ref{katop}).
\end{thm}
\begin{rem}
i) The first statement compares $\cU(t,0)$ with the  Hamiltonian evolution $\cU^0(t,0)$, while the second one uses a leading order approximation of the latter, hence the different error terms. \\
ii) The projector $\cP_0(t)$ on the left of the  integral term in (\ref{apupr}) shows that the coherences of the correction due to the dissipator vanish to leading order. Actually, as soon as $\sigma_{j}(t)=\{e_j(t)\}$,  we have  for any $n\neq m$ in the regime $g\ll \eps\ll 1$, 
\be
P_n(t)\cU(t,0)(\rho_j)P_m(t)=P_n(t)\cU^0(t,0)(\rho_j)P_m(t)+\ode(g+(g/\eps)^{2}),
\ee
see Lemma \ref{vancor2}. \\
iii) Moreover, under the assumptions of Theorem \ref{asuper},  Proposition \ref{propuread} shows that $P_n(t)\cU^0(t,0)(\rho_j)P_m(t)$ is of order $\eps^2$ if $n$ and $m$ are different from $j$, while it is of order  $\eps$ if $n$ or $m$ equals $j$. 
\end{rem}

\subsection{Slow drive regime $\eps \ll g\ll 1$}

We consider now larger time scales $1/\eps$ for the drive, which implies the adiabatic dynamics within the instantaneous eigenspaces of the driving Lindbladian $\cL_t^{[g]}$ will dominate. 
In order to tackle this regime for $g$ small, we will need more precise spectral information on $\cL_t^{[g]}$ that require working in a simpler setup. In particular, the next result  holds under the assumption that $\cH$ is finite dimensional and that $\sigma(H(t))$ is generic in the following sense:

\medskip

\noindent{\bf Gen} \\
\noindent
$\bullet$ $\dim \cH=d$, $1<d<\infty$.\\
$\bullet$ $\forall t\in [0,1]$, $\sigma(H(t))=\{e_1(t),\cdots,e_d(t)\}$ is simple and the Bohr frequencies $\{e_j(t)-e_k(t)\}_{1\leq j\neq k\leq d}$ are distinct.\\

Note that {\bf Spec} holds with $\sigma_j(t)=\{e_j(t)\}$ if {\bf Gen} is satisfied.\\

\noindent
Let $\{\ffi_j\}_{1\leq j\leq d}$ be a fixed orthonormal basis of eigenvectors of $H(0)$. We consider $\{\ffi_j(t)\}_{1\leq j\leq d}$, the orthonormal basis of smooth eigenvectors of $H(t)$ 
defined by 
\be\label{instev}
\ffi_j(t)=W(t,0)\ffi_j, \ \ \mbox{s.t.}  \ \ \bra \ffi_j(t) | \ffi_j'(t)\ket\equiv 0,
\ee
see (\ref{katop}) and (\ref{interk}).
Therefore
\be
P_j(t)=|\ffi_j(t)\ket\bra \ffi_j(t)| \ \ \mbox{and}\  \ H(t)=\sum_{1\leq j\leq d}e_j(t)P_j(t).
\ee  

As a consequence of assumption {\bf Gen}, $\cL_t^0$ admits $0$ as a $d-$fold eigenvalue, with
\be
\ker \cL_t^0 =\spa \{P_j(t),1\leq  j \leq d\},
\ee
whereas all its other eigenvalues are purely imaginary and simple. The spectral projector onto $\ker \cL_t^0$ is the projector $\cP_0(t)$ introduced in (\ref{projoker}).
The splitting of the eigenvalue $0$ of $\cL_t^0$ by the addition of the dissipator $g\cL_t^1$ is thus governed to leading order in $g$ by the operator on $\cB(\cH)$
\be
\tilde \cL_t^1=\cP_0(t)\circ \cL_t^1\circ \cP_0(t).
\ee
See (\ref{mattildel}) for the matrix form of $\tilde \cL_t^1|_{\ker \cL_t^0}$ in the ordered basis $\{P_1(t), \dots, P_d(t)\}$.\\

We assume that the splitting induced by $g\cL_t^1$ is maximal in the following sense

\medskip 
\noindent {\bf Split}  

\noindent
$\bullet$ For all $t\in[0,1]$, the spectrum of the restriction of $\tilde \cL^1_t$ to ${\ker \cL_t^0}$ is simple.\\

\noindent
{ Hypothesis {\bf Split} is generic in the sense that in absence of very specific symmetry, the dissipator $\cL_t^1$ fully lifts the degeneracy of the eigenvalue zero of $\cL_t^0$, which is equivalent to saying that the matrix {representation of $\tilde \cL_t^1|_{\ker \cL_t^0}$}  (\ref{mattildel}) has simple eigenvalues. {  It implies that for each fixed $t\in [0,1]$ and $g>0$ small enough, $\cL_t^{[g]}$ has one dimensional kernel, so that  the (Ces\'aro) limit as $s\ra \infty$ of the semigroup ${e^{s\cL_t^{[g]}}}$ converges to the projector onto that kernel}. This is true for the example worked out in  Section \ref{ex}, and for the simple quantum reset models considered in \cite{HJ}, section 2.1, in particular.}\\

\noindent
Consequently, for $g>0$ small enough, the spectrum of the  Lindbladian $\cL_t^{[g]}$ is simple, see \cite{K2}, with a non trivial kernel. Moreover, thanks to Remark \ref{gershg} below, the real parts of all the $d-1$ eigenvalues of order $g$ are strictly negative for $g>0$ small enough.
The simplicity of the spectrum of $\cL_t^{[g]}$ for small  $g>0$ allows us to construct a propagator $({\cal V}(t,s))_{0\leq s\leq t\leq 1}$, which possesses the intertwining property with  all spectral projectors of  $\cL_t^{[g]}$, and approches  $({\cal U}(t,s))_{0\leq s\leq t\leq 1}$ under the sole condition $\eps\ll g\ll 1$:  it satisfies for $g>0$ small enough and all $0\leq s\leq t\leq 1$  
\be\label{glimpse}
\|\cU(t,s)-{\cal V}(t,s)\|_\tau=\ode(\eps/g).
\ee
The explicit description of ${\cal V}(t,s)$, which depends on the spectral data of $\cL_t^{[g]}$ and $\eps$, is too involved for this presentation section, and we refer the reader to Proposition \ref{je<g} for more details. 

We present  a statement  that holds under the supplementary condition $\eps\ll g \ll \sqrt{\eps.}$
Our second result describes the leading order of the density matrix $\cU(t,0) (P_j(0))$, which is characterised by $\tilde \cL^1_t$,  in the slow drive regime.
\begin{thm}\label{betcont}
Assume {\bf Reg}, {\bf Gen} and {\bf Split}.  Then,  for any fixed  $ 0<t\leq 1$, and  $j\neq k$, the solution to (\ref{ulind}) satisfies for $(\eps, g)\ra (0,0)$ with $\eps/g\ra 0$ and $g^2/\eps\ra 0$, 
\be
\cU(t,0) (P_j(0))=\tilde \nu_0(t)+\ode(g^2/\eps+\eps/g),
\ee
where $\tilde \nu_0(t)=\cP_0(t)(\tilde \nu_0(t))$ is determined by $\tilde \cL_t^1(\tilde \nu_0(t))=0$ and $\tr  (\tilde \nu_0(t))=1$.
\end{thm}
In physical terms, the statements above mean that in the regime considered, the drive is so slow that the dissipator has time enough to determine the instantaneous invariant state that the adiabatic dynamics selects, within the kernel of the Hamiltonian part of the Lindbladian.
\begin{rem}
i) The extra constraint $g^2\ll \eps$ stems from the fact that we only retain $\tilde \cL_t^1$ in the description of  $\cU(t,0) (P_j(0))$.\\
ii) Accordingly, for any $1\leq k\leq d$, the transition probability to $P_k(t)$ starting at $P_j(0)$ reads
\be
\tr (P_k(t) \cU(t,0) (P_j(0)))=\tr(P_k(t)\tilde \nu_0(t))+\ode(g^2/\eps+\eps/g),
\ee
while the coherences all vanish to leading order, since $\tilde \nu_0(t)=\cP_0(t)(\tilde \nu_0(t))$.\\
iii) More precise results taking into account the other eigenstates of $\tilde \cL^1_t$ can be found in Corollary \ref{slowdrive}.\\
\end{rem}

\subsection{Transition regime $g\ll \sqrt \eps\ll 1$}

The asymptotic expressions stated above require either $g\ll \eps$ or $\eps \ll g$, and thus do not cover the transition regime where $\eps$ and $g$ are roughly of the same order. Our next result  bridges this gap: for an initial state $\rho_j=P_j(0)\rho_jP_j(0)\in \cT(\cH)$, it provides an approximation of the evolved state $\cU(t,0) (\rho_j)$ which holds as soon as $g\ll \sqrt \eps\ll 1$, a regime that covers partly both the perturbative and slow drive regimes. Moreover, this result only requires the spectral assumption {\bf Spec} with $\sigma_j(t)=\{e_j(t)\}$, for all $1\leq j\leq d$, regardless of the dimensions of the spectral projectors $P_j(t)$ of $H(t)$.

\medskip

Let $\tilde \cL_t^1=\cP_0(t) \cL_t^1 \cP_0(t)$ and consider $(\tilde \Psi_{\delta}(t,s))_{0\leq s\leq t\leq 1}$, defined for $\delta>0$ by
\bea\label{auxad}
\left\{\begin{matrix}
\delta \partial_t \tilde \Psi_{\delta}(t,s)={\cW_0}(0,t)\tilde \cL_t^1\cW_0(t,0)\tilde \Psi_{\delta}(t,s),\\
 \tilde \Psi_{\delta}(s,s)=\un.\hfill \end{matrix}\right.
\eea
We call $\tilde \Psi_{\delta}(t,s)$ the {\it reduced dynamics}, since (\ref{intertw0}) implies $ [\tilde \Psi_{\delta}(t,s),\cP_0(0)]\equiv 0$ .\\

The following Theorem shows that $\tilde \Psi_{\delta}(t,s)$ with $\delta = \eps/g$ provides an approximation of  $\cU(t,0) \cP_0(0)$ in the transition regime $g\ll \sqrt\eps\ll 1$:
\begin{thm}\label{thmtra} Assume {\bf Reg} and {\bf Spec}  with $\sigma_j(t)=\{e_j(t)\}$ for all $1\leq j\leq d$ and all $t\in[0,1]$. 
Then, for all $0\leq t\leq 1$, as $(\eps, g)\ra 0$ with $g^2/\eps\ra 0$, we have for the solution of (\ref{ulind})
\be
\cU(t,0)\cP_0(0)=\cW_0(t,0)\tilde \Psi_{\eps/g}(t,0)\cP_0(0)+\ode(\eps+g+g^2/\eps).
\ee
Consequently, for any state $\rho_j=P_j(0)\rho_jP_j(0)\in \cT(\cH)$ and  any $1\leq j,k\leq d$, $0\leq t\leq 1$, we have in the same regime,
\begin{align}
\tr \{P_k(t)\cU(t,0)(\rho_j)\}
&=\tr \{P_k(0) \tilde \Psi_{\eps/g}(t,0)(\rho_j)\}+\ode(\eps+g+g^2/\eps).
\end{align}
\end{thm}
In this transition regime, the dissipator is not strong enough to ensure instantaneous relaxation of the dynamics over the adiabatic time scale, while the coherences are still suppressed by the slow drive.
\begin{rem} 
i) The coherences vanish to leading order since
\be
\cW_0(t,0)\tilde \Psi_{\eps/g}(t,0)\cP_0(0)=\cP_0(t)\cW_0(t,0)\tilde \Psi_{\eps/g}(t,0)\cP_0(0).
\ee 
ii) The result holds in particular for $g=\eps$ in which case we have
\be
\cU(t,0)(\rho_j)=\cW_0(t,0)\tilde \Psi_{1}(t,0)(\rho_j)+\ode(\eps),
\ee
where the reduced dynamics is parameter free and thus of order 1.\\
iii) For any $\delta>0$, the maps $\tilde \Psi_{\delta}(t,0)\cP_0(0)$ and $\cW_0(t,0)\tilde \Psi_{\delta}(t,0)\cP_0(0)$ are CPTP, see Corollary \ref{cor1}.\\
iv) In case $\eps\ll g\ll \sqrt \eps$, the reduced dynamic is itself in an adiabatic regime in the parameter $\delta= \eps/g$, which allows us to recover the slow drive regime as shown in Corollaries \ref{corevolad}, \ref{slowdrive}, further assuming {\bf Gen} and {\bf Split}. Whereas for $g\ll \eps\ll 1$ we fall back on the perturbative regime, see Corollary \ref{recovpereg}, under the present hypotheses.\\
v) In Section \ref{ex},  the reduced dynamics $\tilde \Psi_{\eps/g}(t,0)$ is computed explicitly for a two-level system, under mild symmetry assumptions on the jump operators $\Gamma_l(t)$.   \\
vi) Finally, under {\bf Gen}, the reduced dynamics can be interpreted as the transition matrix of an associated classical Markov process, see Lemma \ref{markov}.

\end{rem}

\medskip

The rest of the paper is organised as follows. In the next section, we consider the perturbative regime, making use of Dyson series. To study this series,  we revisit methods in the adiabatic analysis of evolution equations, which leads to Proposition \ref{propweak},  the main technical result of this section. The statements of Theorems \ref{maing<e} and \ref{asuper} are leading order consequences of this result, as explained at the very end of that section. Section \ref{slodri} is devoted to the slow drive regime, starting with the spectral analysis of the Lindbladian for $g$ small, to get the approximation (\ref{glimpse}) of the Lindbladian evolution for $\eps\ll g$ stated as Proposition \ref{je<g},  which is proven there. The transition regime is finally addressed in Section \ref{trareg}, where the reduced dynamics is introduced and analysed. In particular, various asymptotic values of $\delta=\eps/g$ allowed by the condition $g\ll \sqrt \eps$ are considered as Corollaries \ref{cor1}, \ref{corevolad}, \ref{slowdrive} and \ref{recovpereg}   of Proposition \ref{tech}, the main technical result of the section. For instance, Theorem \ref{thmtra} corresponds to Corollary \ref{cor1}, while Corollaries \ref{corevolad} and \ref{slowdrive} are shown to yield the statements of Theorem \ref{betcont}, whereas Corollary \ref{recovpereg}  partially recovers results in the perturbative regime. An application to a two-level system, or Qubit, illustrating those results is worked out in Section \ref{ex}, while Section \ref{secgen} is devoted to higher order generalisations of the perturbative regime results. The paper closes with a technical appendix gathering some proofs.

 \section{Perturbative regime $g\ll \eps$}\label{weakcoup}
 
In the regime $g\ll \eps$, the dissipator of the Lindbladian can be considered a perturbation of the Hamiltonian part, so that a head on approach using Dyson series in the interaction picture is useful. We work here under {\bf Reg} and {\bf Spec}
\medskip
 
Let $(\cU^0(t,s))_{(t,s)\in \R^2}$ be the propagator on $\cT(\cH)$ solution to the equation in $\cT(\cH)$
\begin{align}\label{u0}
\left\{\begin{matrix}
\eps \partial_t\cU^0(t,s)=\cL_t^0(\cU^0(t,s)), \hfill\cr
\cU^0(s,s)=\un, \ \ (t,s)\in [0,1]^2.
\end{matrix} \right.
\end{align}
Introducing $(U(t,s))_{(t,s)\in\R^2}$, the unitary Schr\"odinger propagator, solution to the equation in $\cH$
\be\label{schr}
\left\{\begin{matrix}
\i \eps \partial_t U(t,s)=H(t)U(t,s), \hfill \cr
U(s,s)=\un, \ \ (t,s)\in [0,1]^2,
\end{matrix} \right.
\ee
we check that for any $\rho\in \cT(\cH)$
\be\label{u0u}
\cU^0(t,s)(\rho)=U(t,s) \rho U^*(t,s),
\ee 
showing that $\cU^0(t,s)$ is actually { unitarily implemented}  on $\cT(\cH)$ and on $\cB(\cH)$, and is well defined for any $0\leq s,t\leq 1$. 
 
The integral form of (\ref{ulind})  and the definition of $\cU^0(t,s)$ yield
\be\label{uint}
\cU(t,r)=\cU^0(t,r)+\frac{g}{\eps}\int_r^t \cU^0(t,s)\circ \cL^1_s\circ \cU(s,r)ds, \ \ \forall 0\leq r\leq t \leq 1.
\ee
By iteration we have  for $0\leq s\leq t \leq 1$,  $N\geq 1$, and with the convention $s_0=t$,
\begin{align}\label{dyson}
\cU(t,s)-\cU^0(t,s)&\\
=\sum_{n=1}^N(g/\eps)^n&\int_s^t\int_s^{s_1}\dots\int_s^{s_{n-1}} \cU^0(t,s_1)\circ \cL^1_{s_1}\circ \cU^0(s_1,s_2)\circ \cL^1_{s_2}\dots \circ \cL^1_{s_n}\circ \cU^0(s_n,s) ds_{n}\dots ds_{2}ds_{1} \nonumber\\
+(g/\eps)^{N+1}&\int_s^t\int_s^{s_1}\dots\int_s^{s_{N}} \cU^0(t,s_1)\circ \cL^1_{s_1}\circ \cU^0(s_1,s_2)\circ \cL^1_{s_2}\dots \circ \cL^1_{s_{N+1}}\circ \cU(s_{N+1},s) ds_{N+1}\dots ds_{2}ds_{1}\nonumber\\
=\sum_{n=1}^\infty(g/\eps)^n&\int_s^t\int_s^{s_1}\dots\int_s^{s_{n-1}} \cU^0(t,s_1)\circ \cL^1_{s_1}\circ \cU^0(s_1,s_2)\circ \cL^1_{s_2}\dots \circ \cL^1_{s_n}\circ \cU^0(s_n,s) ds_{n}\dots ds_{2}ds_{1}.\nonumber
\end{align}
The convergence is in the norm operator sense on $\cT(\cH)$ for the second expression. In particular, with $\sup_{0\leq s\leq 1}\|\cL^1_s\|_\cT=L_1$, the norm of the term of order $n$ is bounded above by 
$((g/\eps)L_1(t-s))^n/n!$, since $\cU^0(t,s)$ is isometric.
\begin{rem}
The Dyson series converges, irrespectively of the value of the ratio $g/\eps$.
\end{rem}

Thanks to (\ref{u0u}), the adiabatic approximation of  $\cU^0(t,s)$ is easily obtained from that of the Schr\"odinger propagator $U(t,s)$, under the spectral hypotheses {\bf Spec} on the Hamiltonian.

\subsection{Adiabatic toolbox}
The gap hypothesis {\bf Spec} and the assumed regularity in time of $H(t)$ ensure the existence of a unitary propagator on $\cH$, $(V(t,s))_{(t,s)\in [0,1]^2}$ defined by 
\be\label{adev}
\left\{\begin{matrix}
\i \eps \partial_t V(t,s)=(H(t)+\i \eps K(t))V(t,s),\\ V(s,s)=\un, \ \ (t,s)\in [0,1]^2,  \hfill 
\end{matrix}\right.
\ee
where $K(t)$ is given in (\ref{multik}).
The adiabatic theorem of quantum mechanics reads, see {\it e.g. }\cite{K1, N1, ASY},
\begin{lem}\label{atqm} Under {\bf Reg} and {\bf Spec}, there exists $c$ such that for any $0\leq s, t\leq 1$, and $\eps >0$, the solutions to (\ref{schr}) and (\ref{adev}) satisfy
\begin{align}\label{adiab0}
&\|U(t,s)-V(t,s)\|\leq c\eps,
\end{align}
where $V$ possesses the intertwining property 
\be\label{inter}
V(t,s)P_j(s)=P_j(t)V(t,s), \ \forall 1\leq j\leq d. 
\ee
\end{lem}
\begin{rem}\label{ctoct} Thanks to the second point of assumption {\bf Reg}, we have for $s=0$
\be
\|U(t,0)-V(t,0)\|\leq ct\eps.
\ee
\end{rem}
As a direct consequence, the transition amplitude from the subspace $P_j(s)\cH$ at time $s$, to the subspace $P_k(t)\cH$ at time $t$, $j\neq k$, is of order $\eps$, as $\eps\ra 0$:
\be
\|P_k(t)U(t,s)P_j(s)\|=\|P_k(t)V(t,s)P_j(s)\|+\ode(\eps)=\ode(\eps).
\ee 

Estimate (\ref{adiab0}) in Lemma \ref{atqm} is based on an integration by parts procedure that will be used below in various situations, and even generalised in Section \ref{secgen}. Therefore, the argument is presented in Appendix in a general abstract setup as Lemma \ref{idint}, from which the proof of (\ref{adiab0}) follows. The classical intertwining property (\ref{inter}), obtained by observing that both sides are solutions to the same differential equation in $t$, with initial condition $P_j(s)$ at $t=s$ can be found in (\cite{K1, K2, Kr}).
For later purposes, we also introduce here a useful decomposition of the operator $V(t,s)$:

Let $W(t,s)$ be the Kato operator defined by (\ref{katop})
and 
 $\Phi_\eps(t,s)$ be the {\it dynamical phase} operator defined by
\bea
\left\{\begin{matrix}
\i\eps\partial_t \Phi_\eps(t,s)=W^{-1}(t,0)H(t)W(t,0)\Phi_\eps(t,s), \\
\Phi_\eps(s,s)=\mathbb I, \ \ 0\leq s,t \leq 1.\hfill 
\end{matrix}\right.
\eea
The dynamical phase operator $\Phi_\eps(t,s)$ describes the evolution within the spectral subspaces of $H(t)$, and thus depends on $\eps$. As the Kato operator, it is well defined in a Banach space framework for bounded generator, and its key property is that for all $1\leq j\leq d$
\be
[\Phi_\eps(t,s),P_j(0)]\equiv 0.
\ee
In case {\bf Spec} holds with $\sigma_j(t)=\{e_j(t)\}$, 
$\Phi_\eps(t,s)P_j(0)=P_j(0)\e^{-\frac{\i}{\eps}\int_s^te_j(r)dr}$, and if this assumption holds  for all $1\leq j\leq d$,
\be\label{fidyn}
\Phi_\eps(t,s)=\sum_{j=1}^dP_j(0)\e^{-\frac{\i}{\eps}\int_s^te_j(r)dr},
\ee 
which justifies the name of the operator. 
The link between $V$, $W$ and $\Phi_\eps$ reads, see {\it e.g.} \cite{K1, JP, J2}
\begin{lem}\label{decvwp}
Under {\bf Reg} and {\bf Spec}, one has
\be
V(t,s)=W(t,0)\Phi_\eps(t,s)W^{-1}(s,0), \ \ \forall \ 0\leq t,s\leq 1.
\ee
\end{lem}
Under these assumptions, the operators $V(t,s), W(t,s), \Phi_\eps(t,s)$ are all unitary.
\medskip

In turn, the adiabatic approximation (\ref{adiab0}) $V(t,s)$ of $U(t,s)$  on $\cH$ provides an approximation of $\cU^0(t,s)$ on $\cT(\cH)$ up to $\ode(\eps)$. For $\eps >0$, define the isometric operator on $\cT(\cH)$ (and on $\cB(\cH)$)
\be\label{defv00}
\cV^0(t,s)(\rho)=V(t,s) \rho  V^*(t,s)=V(t,s) \rho  V(s,t), \ \ \rho\in \cT(\cH).
\ee
Then, for $c$ given in equation (\ref{adiab0}), we get  
\bea\label{u0v0sim}
\|\cU^0(t,s)-\cV^0(t,s)\|_\tau\leq 2c\eps,
\eea
and the same holds for the operator norm  on $\cB(\cH)$.

To get a better grasp on $\cV^0(t,s)$,  we proceed by proving here (\ref{par0pdimstate}) which specifies the action of $\cW_0(t,s)$ on the range of $\cP_0(s)$:
\begin{lem}\label{lemw}
For  $\cW_0(t,s)$ defined by (\ref{w0}) and $\cP_0(t)$ by (\ref{projoker}), we have for any $\rho \in \cB(\cH)$,
\be\label{par0pdim}
\cW_0(t,s)\circ \cP_0(s)(\rho)=W(t,s)  \cP_0(s)(\rho) W(s,t),
\ee 
where $W(t,s)$ is the Kato operator defined by (\ref{katop}). Moroever, $\cW_0(t,s)$ is trace preserving.
\end{lem}
\begin{rem}\label{remwp} i) The particular case $\rho=P_j(s)$ which gives $\cW_0(t,s)(P_j(s))=P_j(t)$ can be found \cite{AFGG}.\\
ii) If the projectors are all one dimensional, {\it i.e.} $P_j(t)= | \ffi_j(t)\ket\bra \ffi_j(t)|$, $1\leq j \leq d$,
we have 
\bea\label{par0p}
\cW_0(t,s)\circ \cP_0(s)(\rho)=\sum_{1\leq k\leq d} \bra \ffi_k(s)|\rho \ffi_k(s) \ket P_k(t).
\eea
\end{rem}
\proof 
The fact that $\cW_0(t,s)$ is trace preserving follows from  $\tr ((\un -\cP_0(t))(A))=0$ for all $A\in \cT(\cH)$, and (\ref{intertw0}), together with  (\ref{par0pdim}) and the fact that $\cP_0(s)$ is trace preserving.  The identity is proven by checking that its two sides  satisfy the same differential equation (\ref{w0}) with initial condition $\cP_0(s)$ at $t=s$. Considering first the argument $P_k(s)\rho P_k(s)=\cP_0(s)(P_k(s)\rho P_k(s))$ in place of $\rho$ in (\ref{par0pdim}), the RHS reads
$W(t,s)  P_k(s)\rho P_k(s) W(s,t)=P_k(t) W(t,s) \rho W(s,t)P_k(t)$ so that with (\ref{multik}),
\begin{align}
\partial_t \Big(W(t,s)  P_k(s)\rho P_k(s) W(s,t)\Big)=&[K(t), W(t,s)  P_k(s)\rho P_k(s) W(s,t)]\\ \nonumber
=&P_k'(t)P_k(t) W(t,s) \rho W(s,t)P_k(t)+ P_k(t)W(t,s) \rho W(s,t)P_k(t)P_k'(t).
\end{align}
Then, making use of 
\begin{align}
 \cP_0'(t)(\rho)&=\sum_{1\leq j\leq d}P_j'(t)\rho P_j(t)+P_j(t)\rho P_j'(t),
\end{align}
we get $\cP_0(t)\cP_0'(t)(P_k(t) W(t,s) \rho W(s,t)P_k(t))\equiv 0$ and 
\begin{align}
\cP_0'(t)\cP_0(t)(P_k(t) W(t,s) \rho &W(s,t)P_k(t))\\ \nonumber
&=P_k'(t)P_k(t) W(t,s) \rho W(s,t)P_k(t)+ P_k(t)W(t,s) \rho W(s,t)P_k(t)P_k'(t),
\end{align}
showing the result for the argument $P_k(s)\rho P_k(s)$. It remains to sum over $1\leq k\leq d$ to end the proof.
\qed\\
As a consequence, we get the following expression for the adiabatic approximation (\ref{u0v0sim}) of  $\cU^0(t,s)(\rho_j)$
where the state $\rho_j=P_j(s)\rho P_j(s)$ and  $P_j(s)$ is associated to a permanently degenerate eigenvalue:
\begin{lem}\label{expadiab0}
Under {\bf Reg} and {\bf Spec} with  $\sigma_j(t)=\{e_j(t)\}$ for all $t\in[0,1]$, 
for any state $\rho_j=P_j(s)\rho P_j(s)=\cP_0(s)(\rho_j)$ it holds
\begin{align}
&\cU^0(t,s)(\rho_j)=\cV^0(t,s)(\rho_j)+\ode(\eps), \ \ \mbox{where}\nonumber\\
&\cV^0(t,s)(\rho_j)=\cW_0(t,s)(\rho_j)=W(t,s)\rho_jW(s,t).
\end{align}
\end{lem}
\begin{rem}\label{v=w} If $\sigma_j(t)=\{e_j(t)\}$ for all $1\leq j\leq d$ and all $t\in [0,1]$, then for any $0\leq s\leq t\leq 1$,
\be\label{vpw}
\cP_0(t)\cW_0(t,s)\cP_0(s)=\cP_0(t)\cV^0(t,s)=\cV^0(t,s)\cP_0(s).
\ee
\end{rem}
\proof The first estimate is (\ref{u0v0sim}). Then, Lemma \ref{decvwp} yields the expression 
\be
V(t,s)=W(t,0)\Phi_\eps(t,s)W(0,s) \ \mbox{ with} \  \ \Phi_\eps(t,s)P_j(0)=P_j(0)\e^{-\frac{\i}{\eps}\int_0^te_j(r)dr},
\ee 
since $\sigma_j(t)=\{e_j(t)\}$. Therefore, $\cV^0(t,s)$ given by (\ref{defv00}) with $\rho_j=P_j(s)\rho_jP_j(s)$ and the intertwining relation (\ref{interk}) make
the phases $\e^{\pm\frac{\i}{\eps}\int_s^te_j(r)dr}$ disappear, which justifies the last two identities. 
\qed\\
In particular, if $\dim P_j(0)<\infty$, then $P_j(t)$ belongs to $\cT(\cH)$ for all $t\in[0,1]$ so that 
\bea\label{superad2}
\cU^0(t,s)(P_j(s))&=&P_j(t)+\ode(\eps).
\eea
Again, if $P_j(0)$ is not trace class, the estimates above hold in operator norm.
\medskip

\subsection{Adiabatic Dyson expansion}\label{return}

We now apply the foregoing to the analysis of the Dyson series.

Equation (\ref{dyson}) and the above yields the following estimate of the propagator $(\cU(t,s))_{0\leq s\leq t\leq 1}$:
\begin{prop}\label{propweak} Under  assumptions {\bf Reg} and {\bf Spec}, for any $N\geq 1$,  there exists  $c<\infty$ (given in Lemma \ref{atqm}), such that for all $0\leq s\leq t \leq 1$ (with the convention $s_0=t$ {  for $N=1$}),
for all $\eps>0$, all $g\geq 0$, the propagator 
$\cU(t,s)\in \cB(\cT(\cH))$ satisfies  
\begin{align}\label{expgenew}
 \cU(t,s)&=\cV^0(t,s)\nonumber \\
 &+\sum_{n=1}^N(g/\eps)^n\int_s^t\int_s^{s_1}\dots\int_s^{s_{n-1}} \cV^0(t,s_1)\circ \cL^1_{s_1}\circ \cV^0(s_1,s_2)\circ \cL^1_{s_2}\dots \circ \cL^1_{s_n}\circ \cV^0(s_n,s) ds_{n}\dots ds_{2}ds_{1} \nonumber\\
 &+R_{N+1}(t,s,\eps,g)
 \end{align}
 where, with $L_1=\sup_{0\leq s\leq 1}\| \cL_s^1\|_\tau$,
 \begin{align}\label{1stnew}
\|R_{N+1}(t,s,\eps,g)\|_\tau\leq 2 c\eps\e^{2(t-s)L_1(1+2 c\eps)g/\eps}+\frac{(L_1(t-s))^{N+1}}{(N+1)!}(g/\eps)^{N+1}.
\end{align}
In particular, if $g/\eps\leq 1$ and $\eps\leq 1/(2c)$,
 \begin{align}
\|R_{N+1}(t,s,\eps,g)\|_\tau&\leq  2\e^{4L_1}\left(c\eps+((t-s)g/\eps)^{N+1}\right)\nonumber\\
&=\ode(\eps+(g/\eps)^{N+1}).
\end{align}
\end{prop}
\begin{rem}\label{remu0v0new}
i) The isometric operators $\cV^0$ depends on $\eps$ and displays fast oscillations as $\eps\ra 0$.\\
ii) Keeping $\cU^0$ instead of $\cV^0$ in the first term of the RHS of (\ref{expgenew}),  the  estimate on the remainder reads
\be
\|R_{N+1}(t,s,\eps,g)\|_\tau=\ode(g+(g/\eps)^{N+1})).
\ee
iii) In case $s=0$, on can replace $c$ by $tc$, according to Remark \ref{ctoct}.
\end{rem}
\proof 
The first estimate follows by replacing $\cU^0$ by its approximation $\cV^0$  in each term of the Dyson series, and collecting the different contributions to the error terms. With 
\be
\Delta=\sup_{0\leq s\leq t\leq 1}\|\cU^0(t,s)-\cV^0(t,s)\|_\tau\leq 2c\eps,
\ee
 the trace norm of the difference of the term of order $n\geq 1$ in (\ref{dyson}) with that of order $n$ in (\ref{expgenew})  is bounded above by
\begin{align}\label{basesdya}
&(g/\eps)^n\frac{((t-s)L_1)^n}{n!}\sum_{1\leq k\leq n+1}\begin{pmatrix} n+1 \\ k \end{pmatrix}\Delta^k\leq (g/\eps)^n\frac{((t-s)L_1)^n}{n!} \sum_{0\leq j\leq n}\begin{pmatrix} n+1 \\ j+1 \end{pmatrix}\Delta^{j+1}\nonumber\\
&\leq (g/\eps)^n\frac{((t-s)L_1)^n}{n!}\Delta (n+1)\sum_{0\leq j\leq n}\begin{pmatrix} n \\ j \end{pmatrix}\Delta^j\leq (g/\eps)^n\frac{((t-s)L_1)^n}{n!}\Delta 2^n(1+\Delta)^{n}.
\end{align}
Summing over all $n\in \N$ yields the first term in (\ref{1stnew}). The second term stems from the term of order $N+1$ in (\ref{dyson}).
 The second estimate is a consequence of  $g/\eps \leq 1$, $t-s\leq 1$, $1+2 c\eps \leq 2$ and  $\frac{\alpha^{m}}{m!}\leq e^{\alpha}$, for all $m\geq 1$, $\alpha>0$.
\qed

Specialising to the leading order term in $g/\eps$, and taking into account Remark \ref{remu0v0new} ii) above, we get 
\begin{cor}\label{coruv} Under the assumptions of Proposition \ref{propweak},  for $\eps\leq 1/(2c)$ and  $g/\eps\leq 1$, 
\begin{align}\label{ordeg/e}
 \cU(t,s)&=\cU^0(t,s)+\frac{g}{\eps}\int_s^t\cV^0(t,s_1)\circ \cL^1_{s_1}\circ \cV^0(s_1,s) ds_{1} +\ode(g+g^2/\eps^{2}).
 \end{align}
\end{cor}
This expression will provide an explicit leading order correction to the transition probability between spectral subspaces driven by the purely Hamiltonian dynamics, due to the dissipator. 
\medskip

Consider a state $\rho_j=P_j(0)\rho_jP_j(0)\in \cT(\cH)$  
and recall the definition (\ref{lind}) of the dissipator
\be\label{defl1}
\cL^1_t(\cdot)=\sum_{l\in I}\Gamma_l(t)\cdot \Gamma_l^*(t) -\frac12\{\Gamma_l^*(t)\Gamma_l(t), \cdot \}.
\ee
The transition probability between $P_j(0)\cH$ and $P_k(t)\cH$, $j\neq k$, induced by the Lindbladian dynamics (\ref{lindeq}) reads $\tr(P_k(t)\cU(t,0)(\rho_j))$. 
 Using (\ref{defv00}) and Lemma \ref{atqm}, we have
 \be\label{rhofi}
 \tilde\rho_j(t,\eps)=\cV^0(s,0) (\rho_j)=V(s,0) \rho_j V(s,0)=P_j(s) V(s,0)\rho_j V(0,s)P_j(s),
 \ee
so that with $P_j(t)P_k(t)=0$ and the cyclicity of the trace,  
 \begin{align}\label{traproqnew}
 \tr(P_k(t)\cU(t,0)&(\rho_j))= \tr(P_k(t)\cU^0(t,0)(\rho_j))\\ 
 &+\frac{g}{\eps}\int_0^t \tr (P_k(t)\cV^0(t,s)\circ \cL^1_{s}\circ \cV^0(s,0) (\rho_j))ds +\ode(g+g^2/\eps^{2})\nonumber\\
&\phantom{xxxx}= \tr (P_k(t)U(t,0)P_j(0)\rho_jP_j(0)U(0,t)P_k(t))\nonumber \\
 &+\frac{g}{\eps}\sum_{l\in I}\int_0^t \tr (P_k(s)\Gamma_l(s) V(s,0)\rho_j V(0,s)\Gamma_l^*(s)P_k(s))ds+\ode(g+g^2/\eps^{2}).\nonumber
 \end{align}
The first expression on the RHS yields the Hamiltonian adiabatic transition probability between these subspaces, whereas the non-negative second term of order $g/\eps$ describes the effect of the environment. 
Note that in case $P_j(0)$ is finite rank, choosing $\rho_j=P_j(0)/\dim(P_j(0))$ yields the simpler  integrands 
\be
 \tr (P_k(s)\Gamma_l(s)P_j(s)\Gamma_l^*(s)P_k(s))/\dim(P_j(0)),
 \ee
and  by Lemma \ref{expadiab0}, $V(s,0)\rho_j V(0,s)=\tilde\rho_j(s)$ is independent of $\eps$ if $\sigma_j(t)=\{e_j(t)\}$.
\medskip

Concerning coherences of the integral term in (\ref{ordeg/e}), we have the following integration by parts result, whose proof is given in Appendix. 
\begin{lem}\label{vancor2} Assume {\bf Reg}, {\bf Spec} and let $\rho_j=P_j(0)\rho_j P_j(0)$ be a state. Suppose $\sigma_j(t)=\{ e_j(t)\}$ for all $t\in [0,1]$.
Then, for any $1\leq n\neq m\leq d$, and all $\eps>0$,
\be\label{cohva}
\frac{g}{\eps}P_n(t)\int_0^t\cV^0(t,s)\circ \cL^1_{s}\circ \cV^0(s,0) (\rho_j)ds\, P_m(t)=\ode (g).
\ee
\end{lem}

Actually, a similar result holds for each term in (\ref{expgenew}). For simplicity, we choose to express it under the supplementary condition  $\sigma_k(t)=\{e_k(t)\}$ for all $1\leq k\leq d$, and postpone its proof to the Appendix. It will allow us to make contact with the reduced dynamics later on.
\begin{lem} \label{gendiag}
Assume {\bf Reg} and {\bf Spec} with  $\sigma_k(t)=\{e_k(t)\}$ for all $1\leq k\leq d$ and all $t\in [0,1]$. Then  for all $n\geq 1$, there exists $\eps_{n}>0$ such that for $\eps<\eps_{n}$ and any state $\rho\in \cT(\cH)$ 
(with  $s_0=t$),
\begin{align}\label{diagn}
&\int_0^t\int_0^{s_1}\dots\int_0^{s_{n-1}} \cV^0(t,s_1) \cL^1_{s_1} \cV^0(s_1,s_2) \cL^1_{s_2}\dots  \cL^1_{s_n} \cV^0(s_n,0) (\rho)ds_{n}\dots ds_{2}ds_{1}\nonumber\\
&=\int_0^t\int_0^{s_1}\dots\int_0^{s_{n-1}}\cW^0(t,s_1)\cP_0(s_1)\cL^1_{s_1}
\cdots \cW^0(s_{n-1},s_n)\cP_0(s_n)\cL^1_{s_n} \cW^0(s_{n},0)\cP_0(0)(\rho)ds_{n}\dots ds_{2}ds_{1}\nonumber\\
&\hspace{3cm}+\ode_{n}(\eps).
\end{align}
\end{lem}
Here $\ode_{n}(\eps)$ means a quantity bounded by $C_n\eps$, where $C_n$ depends on $n$.
Consequently, under the assumptions of Lemma \ref{gendiag}, for all $N\geq 1$, in the regime $g\ll \eps \ll 1$, we have the generalisation of Theorem \ref{asuper}
\begin{align}\label{expgenocoh}
 &\cU(t,0)=\cU^0(t,0)\\
 &+\sum_{n=1}^N\Big(\frac{g}{\eps}\Big)^n\hspace{-.15cm}\int_0^t\int_0^{s_1}\dots\int_0^{s_{n-1}}\cW^0(t,s_1)\cP_0(s_1)\cL^1_{s_1}
\cdots \cW^0(s_{n-1},s_n)\cP_0(s_n)\cL^1_{s_n} \cW^0(s_{n},0)\cP_0(0)(\rho)ds_{n}\dots ds_{1} \nonumber\\ \nonumber 
 &+\ode_N(g+(g/\eps)^{N+1}),
 \end{align}
where all integral terms are independent of $(\eps, g)$. Note that the error term in (\ref{expgenocoh}) is negligeable with respect to all explicit terms in the regime $\eps^{N/(N-1)}\ll g\ll \eps$.\\

We proceed by recalling the asymptotics of the transition probability between the spectral subspaces $P_j(0)\cH$ and $P_k(t)\cH$ under the unitarily implemented evolution $\cU^0(t,s)$ and of the coherences of $\cU^0(t,s)(\rho_j)$ depending on certain assumptions on the spectral subsets $\sigma_k(t)$, $1\leq k \leq d$.

\begin{prop}\label{propuread} Assume {\bf Reg} and {\bf Spec} and consider a state
$\rho_j=P_j(0)\rho_j P_j(0)\in \cT(\cH)$. Then, for $\sigma_j(t)=\{e_j(t)\}$, the adiabatic transition probability between $P_j(0)\cH$ and $P_k(t)\cH$, $k\neq j$,  under the Hamiltonian evolution is determined by the trace of 
\begin{align}\label{puread}
  &P_k(t)\cU^0(t,0)(\rho_j)P_k(t)\\
 &=-\frac{\eps^2}{(2\pi)^2} \Big\{P_k(t)\oint_{\gamma_j}R(t,z)P_k'(t)R(t,z)dz \ \tilde\rho_j(t) \oint_{\gamma_j}R(t,z)P_k'(t)R(t,z)dz \, P_k(t)\Big\}+\ode(\eps^3),\nonumber
\end{align}
where $\tilde\rho_j(t)=W(t,0)\rho_jW(0,t)$ and $R(t,z)=(H(t)-z)^{-1}$, for $z\in\rho(H(t))$. 

In case $\sigma_j(t)=\{e_j(t)\}$ and $\sigma_k(t)=\{e_k(t)\}$ for all $t\in [0,1]$,  
we have
\begin{align}\label{pureadev}
 P_k(t)\cU^0(t,0)(\rho_j)P_k(t)
 &=\eps^2 \Big\{\frac{P_k(t)P_k'(t)\tilde\rho_j(t)P_k'(t) P_k(t)}{(e_j(t)-e_k(t))^2}\Big\}+\ode(\eps^3).
 \end{align}
 
Further assuming  $\sigma_k(t)=\{e_k(t)\}$, for all $1\leq k \leq d$, all $t\in [0,1]$, the coherences read
\begin{align}\label{cohmn}
 P_n(t)\cU^0(t,0)(\rho_j)P_m(t)&=\eps^2 \Big\{\frac{P_m(t)P_m'(t)\tilde\rho_j(t)P_n'(t) P_n(t)}{(e_j(t)-e_m(t))(e_j(t)-e_n(t))}\Big\}+\ode(\eps^3),  \phantom{xxxxxx} \mbox{for }  m\neq j, n\neq j, \nonumber\\
 P_j(t)\cU^0(t,0)(\rho_j)P_m(t)&=\i\eps \frac{\tilde \rho_j(t)P_m'(t)P_m(t)}{e_m(t)-e_j(t)} +\ode(\eps^2),\phantom{xxxxxxxxxxxxxxxxxxxxxxx}   \mbox{for }  m\neq j. 
\end{align}
\end{prop}
\begin{rem} \label{remtiltra} 
In case the assumption  $\sigma_j(t)=\{e_j(t)\}$ is dropped, $\tilde \rho_j(t)$ must be replaced by the $\eps-$dependent state 
\be
\tilde \rho_j(t, \eps)=\cV^0(t,0)(\rho_j)=V(t,0)\rho_j V(0,t)=P_j(t)\tilde \rho_j(t, \eps)P_j(t)
\ee 
in (\ref{puread}), see (\ref{defv00}). 
Similar integral expressions can be obtained for the coherences in case condition $\sigma_k(t)=\{e_k(t)\}$ for all $1\leq k\leq d$ does not hold.
\end{rem}
The proof, making use of higher order adiabatic approximations is postponed to Section \ref{secgen}.  
\\

At this point, the proof of Theorem \ref{maing<e} follows from (\ref{traproqnew}) and Proposition (\ref{propuread}): recall that $\sigma_j(t)=\{e_j(t)\}$ implies $V(t,0)\rho_j V(0,t)=W(t,0)\rho_jW(0,t)=\tilde\rho_j(t)$ by (\ref{defv00}) and Lemma \ref{expadiab0}, which yields eq. (\ref{sigjej}) from (\ref{traproqnew}). Then, (\ref{214}) follows directly from (\ref{pureadev}).

Similarly, the proof of the first line of (\ref{apupr}) in Theorem \ref{asuper} is a consequence of Corollary \ref{coruv} for $s=0$, Lemma \ref{vancor2} and Remark \ref{v=w}. The second line follows from  Lemma \ref{expadiab0} and $g \ll \eps$.

 \section{Slow drive regime $g \gg \eps$}\label{slodri}
 
As already mentioned in Section \ref{setup}, the analysis of the regime $g\gg \eps$, is all the more accurate that we control the spectral properties of the Lindbladian $\cL^{[g]}_t=\cL^0_t+g\cL^1_t$. In this section, we provide a fairly explicit approximation of $\cU(t,s)$ based on perturbation theory under the sole condition $\eps \ll g$, see Proposition \ref{je<g}, 
assuming  the Hilbert spaces $\cH$ is finite dimensional and the Hamiltonian $H(t)$ has simple eigenvalues that are generic in the sense of {\bf Gen}. 

\medskip

We recall that for $H(t)=\sum_{1\leq j\leq d}P_j(t)e_j(t)$, the explicit description of the approximate evolution operator $V(t,s)$ defined by (\ref{adev}) provided by Lemma \ref{decvwp} holds with 
\be
\Phi_\eps(t,s)=\sum_{j=1}^dP_j(0)\e^{-\frac{\i}{\eps}\int_s^te_j(r)dr},
\ee 
irrespectively of the dimension of the projectors $P_j(t)$, see (\ref{fidyn}).  
Under {\bf Reg} and {\bf Gen}, we consider $\{\ffi_j(t)\}_{1\leq j\leq d}$, the canonical smooth orthonormal basis  of eigenvectors of $H(t)$ defined in (\ref{instev}) that satisfies
\be\label{instev2}
\ffi_j(t)=W(t,0)\ffi_j, \ \ \mbox{s.t.}  \ \ \bra \ffi_j(t) | \ffi_j'(t)\ket\equiv 0,
\ee
so that $P_j(t)=|\ffi_j(t)\ket\bra \ffi_j(t)|$.
 Introducing the { Hilbert-Schmidt} scalar product on $\cB(\cH)$
 \be\label{BRAKET}
 \BRA A,B\KET=\tr (A^*B),
 \ee
the vectors (\ref{instev2}) yield in turn an instantaneous orthonormal eigenbasis of $\cL_t^0$ with respect to $\BRA \cdot, \cdot \KET$:
\begin{lem} For all $t\in[0,1]$, the family  of rank one operators on $\cH$, $\{|\ffi_j(t)\ket\bra \ffi_k(t)|\}_{1\leq j,k\leq d}$, is a smooth orthonormal basis of $\cB(\cH)$ equipped with $\BRA \cdot, \cdot \KET$ such that
\begin{align}\label{specl0}
&\cL_t^0(|\ffi_j(t)\ket\bra \ffi_k(t)|)=-\i (e_j(t)-e_k(t)) |\ffi_j(t)\ket\bra \ffi_k(t)|:=\lambda_{jk}(t)  |\ffi_j(t)\ket\bra \ffi_k(t)|
\end{align}
\end{lem}
Consequently, the spectrum of $\cL^0_t$ sits on the imaginary axis, is simple except for the eigenvalue $0$ which is $d-$fold degenerate, and 
\be
\ker \cL^0_t = \spa \{P_j(t), j=1,\dots, d\}.
\ee
The corresponding spectral decomposition is written as
\be\label{specdec0}
\cL_t^0=\sum_{j\neq k} \lambda_{jk}(t) \cP_{jk}(t)+ 0 \cP_0(t),
\ee
where
\begin{align}\label{specpro}
\cP_{jk}(t)(\rho)&=P_j(t)\rho P_k(t)=\bra \ffi_j(t)| \rho \ffi_k(t)\ket  |\ffi_j(t)\ket\bra \ffi_k(t)|\nonumber\\
\cP_{0}(t)(\rho)&=\sum_{1\leq j\leq d}{P_j(t)\rho P_j(t)}.
\end{align}
With $\cQ_0(t)=\un-\cP_0(t)$, we have $\cL_t^0=\cQ_0(t) \circ \cL_t^0= \cL_t^0\circ\cQ_0(t)$.

\subsection{Perturbation theory of $\cL_t^{[g]}$}

We now address the spectral properties of $\cL_t^{[g]}=\cL_t^0+g\cL^1_t$ in the perturbative regime $g\ra 0$. Before doing so, we recall that $\cL_t^{[g]}$ being 
a Lindblad operator, this imposes the following structural constraints on its spectrum, for each $t\in[0,1]$: 
\be\label{speclind}
0\in \sigma(\cL_t^{[g]})=\overline{\sigma(\cL_t^{[g]})}\subset\{z\in \C \ | \Re z\leq 0\}, \ \ \forall \ g\in \R^+.
\ee 
Indeed, for fixed $t\in[0,1]$, $(e^{s\cL_t^{[g]}})_{s\geq 0}$ being a contraction semigroup on $\cT(\cH)$, it follows that $\Re\sigma(\cL_t^{[g]})\leq 0$
and that all eigenvalues sitting on the imaginary axis are semisimple, ${i.e.}$ there is no eigennilpotent (Jordan block) corresponding to those eigenvalues in the spectral decomposition of $\cL_t^{[g]}$. Since $e^{s\cL_t^{[g]}}$ is a CPTP map  for all $s\geq 0$, it admits $1$ as an eigenvalue, see e.g. \cite{Schr}, so that $0$ is an eigenvalue of the generator $\cL_t^{[g]}$. The symmetry $ \cL_t^{[g]}(\rho^*)= (\cL_t^{[g]}(\rho))^*$ which holds for all $\rho\in \cT(\cH)$, implies the symmetry of the spectrum, when applied to eigenvectors of $\cL_t^{[g]}$.

Therefore, following 
Chapter II \S 2 \cite{K2} and dropping the variable $t$ from the notation,  we denote by $\cP_{jk}^{[g]}$ the spectral projector of $\cL^{[g]}$ associated with the  eigenvalues emanating from the unperturbed eigenvalue $-\i (e_j-e_k)$, and by $\cP_0^{[g]}$ the spectral projector on the set of eigenvalues emanating from the $d-$fold degenerate unperturbed eigenvalue $0$, the so-called  $\lambda-$group of eigenvalues, with $\lambda=0$. Since $\cP_{jk}$ is one dimensional, $\cP_{j,k}^{[g]}$ is one dimensional and analytic in $g\in \C$, $|g|$ small enough. On the other hand, $\dim \cP_{0}=d$ implies that the degenerate eigenvalue $0$ may split for non zero $g$ and while $\cP_0^{[g]}$ is analytic in $g\in \C$, for $|g|$ small enough, it might not be the case for the projections on the individual eigenvalues emanating from $0$. 

Let $\{\lambda_j^{[g]}\}_{0\leq j\leq m}$, $m\leq d-1$, be the set of eigenvalues in the  $0-$ group for $g\in \C\setminus\{0\}$ with $|g|$ small enough. Each $\{\lambda_j^{[g]}\}_{0\leq j\leq m}$ is an analytic functions of a (fractional) power of $g$ that tend to zero as $g\ra 0$ and  may be permanently degenerate. For the structural reasons recalled above, one of these eigenvalues, we  denote by $\lambda_0^{[g]}$, is identically equal to zero, $\lambda_0^{[g]}\equiv 0$, $\forall g\in\C\setminus\{0\}$. In case $\lambda_0^{[g]}$ is degenerate, it is semisimple. 

Let us analyse the splitting of the $0$-group of eigenvalues.
We have
\be\label{0gpespero}
\cP_0^{[g]}=-\frac{1}{2i\pi}\int_{\gamma_0}(\cL^{[g]}-z)^{-1}dz=\cP_0+g\cP_1
+\ode(g^2), 
\ee 
for $|g|$ small, where  $\gamma_0$ is a circle of small radius, independent of $g$, centered at the origin. 
We mention for completeness, that since $0$ is a semisimple eigenvalue of $\cL^0$, 
\be\label{q1}
\cP_1=- \cP_0\cL^1 \cS_0-\cS_0\cL^1 \cP_0=\cP_0\cP_1\cQ_0+\cQ_0\cP_1\cP_0,
\ee
 where $\cS_0$ is the reduced resolvent of $\cL_0$ at $0$, satisfying $\cS_0\cP_0=\cP_0\cS_0=0$ and $\cS_0\cL^0=\cL^0\cS_0=\cQ_0$. 
  
The analytic operator that describes the splitting  reads
\begin{align}\label{deftl1}
 \cP_0^{[g]}\cL^{[g]}\cP_0^{[g]} 
&=(\cP_0+g\cP_1+\ode(g^2))(\cL^0+g\cL^1)(\cP_0+g\cP_1+\ode(g^2)) \nonumber \\
&= g\cP_0 \cL^1 \cP_0 +  \ode(g^2),
\end{align}
where we used $\cL_0\cP_0=\cP_0\cL_0=0$.

\medskip

The restriction to ${\ker \cL^0}$ of the leading order term $\tilde \cL^1=\cP_0 \cL^1 \cP_0$ admits a matrix representation $\tilde L$ in the ordered orthonormal basis of rank one projectors $\{P_1, P_2, \dots, P_d\}$ of $\cP_0 \cB(\cH)$  whose elements are determined by the expressions (recall (\ref{specpro}) and (\ref{defl1}))
\be
P_k \cL^1(P_j)P_k=\sum_{l\in I} \left(P_k\Gamma_l^*P_j\Gamma_l P_k-\delta_{jk}P_k\Gamma_l^*\Gamma_l P_k\right).
\ee
We set $\tilde L=\sum_{l\in I}\tilde L(l)\in M_d(\R)$, where the matrix elements $\tilde L(l)_{jk}$, $1\leq j,k\leq d$, read 
\begin{align}
\tilde L(l)_{kj}=|\bra\ffi_k|\Gamma_l\ffi_j\ket|^2-\delta_{jk}\|\Gamma_l\ffi_k\|^2,
\end{align}
so that for any $\rho=\cP_0 (\rho)=\sum_{1\leq j\leq d}\rho_j P_j$, with $\rho_j=\bra\ffi_j | \rho \ffi_j\ket$,
\be\label{L1M1}
\tilde \cL^1(\rho)=\cP_0 \cL^1 \cP_0(\rho)=\sum_{1\leq j, k\leq d} P_k \tilde L_{kj}\rho_j.
\ee
In other words, 
\begin{align}\label{mattildel}
\tilde \cL^1|_{\ker \cL^0}\simeq \tilde L=\sum_{l\in I}
\begin{pmatrix}
|\bra\ffi_1|\Gamma_l\ffi_1\ket|^2 - \|\Gamma_l\ffi_1\|^2 & |\bra\ffi_1|\Gamma_l\ffi_2\ket|^2 & & |\bra\ffi_1|\Gamma_l\ffi_d\ket|^2 \cr 
|\bra\ffi_2|\Gamma_l\ffi_1\ket|^2  & |\bra\ffi_2|\Gamma_l\ffi_2\ket|^2 - \|\Gamma_l\ffi_2\|^2& & |\bra\ffi_2|\Gamma_l\ffi_d\ket|^2 \cr
  & & \ddots & \cr
|\bra\ffi_d|\Gamma_l\ffi_1\ket|^2  & |\bra\ffi_d|\Gamma_l\ffi_2\ket|^2 & & |\bra\ffi_d|\Gamma_l\ffi_d\ket|^2 - \|\Gamma_l\ffi_d\|^2
\end{pmatrix}.
\end{align}
Note that the real matrices $\tilde L(l)$ have non negative off diagonal elements and satisfy 
\be\label{trastoch}
\sum_{1\leq j\leq d}\tilde L(l)_{jk}=0, 
\ee
so that the same properties hold for $\tilde L$. This is a reflection of the fact that $\cL^1$ being a lindbladian, we have $\tr (\cL^1(\rho))=0$ for any $\rho\in\cB(\cH)$, where for $\rho=\cP_0(\rho)$, 
\be\label{traprol1}
\tr (\cL^1(\rho))=\tr \{\cP_0 \cL^1 \cP_0(\rho)\}=\sum_{1\leq j, k\leq d} \tr (P_k) \tilde L_{kj}\rho_j=\sum_{1\leq j, k\leq d} \tilde L_{kj}\rho_j.
\ee
In particular $0\in \sigma(\tilde L)$. Sufficient and necessary conditions for  $\ker \tilde L$ to be one dimensional are given in { \cite{N} or  \cite{D2}, Chapter 12.}.

Note that in the language of classical Markov processes, (\ref{trastoch}) makes the transpose of $\tilde L$ a time-dependent transition rate matrix, or generator, of a Markov process, see {\it e.g.} \cite{YZ}. We further comment on this in section \ref{secmarkov}.\\

We suppose that the splitting induced by $\cL^1$ is maximal by assuming {\bf Split}, {\it i.e.} that $\tilde \cL^1_t|_{\ker \cL_t^0}$ has simple spectrum.

\begin{rem}\label{gershg}
Assumption {\bf Split} and Gershgorin Theorem imply that for any $t\in[0,1]$ and $g>0$ small enough, 
$\Re \sigma( \cP_0(t)^{[g]}\cL_t^{[g]}\cP_0(t)^{[g]}){\setminus\{0\}}<0$.
\end{rem}

Under this hypothesis on the efficiency of the dissipator, we have the spectral decomposition
\be\label{sectilde}
\tilde \cL^1_t=0 \tilde \cP_0(t)+\sum_{j=1}^{d-1}\tilde \lambda_j(t) \tilde \cP_j(t),
\ee
where the distinct eigenvalues $\tilde \lambda_j(t)$ and eigenprojectors $\tilde \cP_j(t)$ are smooth in $t\in [0,1]$. \\

Assumption {\bf Split} ensures the spectrum of $\cL_t^{[g]}$ is simple for { small} $g>0$, and its eigenprojectors are all  regular as $g\ra 0^+$, despite $\cL_t^{[g]}$ is not normal:
\begin{prop}\label{regperpro} Assume {\bf Reg}, {\bf Gen}, and {\bf Split}. Then, there exists $g_0>0$ such that for all $t\in [0,1]$,  for all $g\in\C\setminus \{0\}$ with $|g|<g_0$, $\cL_t^{[g]}$ admits  $d$ distinct eigenvalues $\{\lambda_j^{[g]}(t)\}_{0\leq j\leq d-1}$, with corresponding one dimensional eigenprojectors $\cP_{j0}^{[g]}(t)$, that are $C^\infty$ in $t$ and analytic in $g$. Moreover, $\lambda_0^{[g]}(t)\equiv 0$,  and $\lim_{g\ra 0}\lambda_j^{[g]}(t)/g= \tilde \lambda_j(t)$ and $\lim_{g\ra 0}\cP_{j0}^{[g]}(t)=\tilde \cP_j(t)$, where $\{\tilde \lambda_j(t)\}$ and $\tilde \cP_j(t)$ are the spectral data (\ref{sectilde}) of $\tilde \cL_t^1$ .
\end{prop}
\proof
It is a direct consequence of analytic perturbation theory, since for $t\in[0,1]$ fixed (omitted in the notation), the operator $ \frac1g\cP_0^{[g]}\cL^{[g]}\cP_0^{[g]} $ for $g\in\C\setminus \{0\}$ with $|g|<g_0$ admits an analytic extension to $\{g\in \C, |g|<g_0\}$ with term of order $g^0$ given by $\tilde \cL^1$ thanks to (\ref{deftl1}), with $\sigma(\tilde \cL^1|_{\ker \cL^0})$ simple. \qed

\subsection{Adiabatics and perturbation theory}

Under the hypotheses of the previous proposition and for $g\in \C\setminus \{0\}$, $|g|$ small enough, $\sigma(\cL_t^{[g]})$ is simple and its spectral decomposition reads
\be\label{specdeclg}
\cL_t^{[g]}=0\cP_{00}^{[g]}(t)+\sum_{1\leq j\leq d-1} \lambda_j^{[g]}(t)\cP_{j0}^{[g]}(t)+\sum_{1\leq j\neq k\leq d}\lambda_{jk}^{[g]}(t)\cP_{jk}^{[g]}(t),
\ee
with analytic data in $g$, where  $\lambda_{jk}^{[0]}(t)=-\i (e_j(t)-e_k(t))$, $\cP_{jk}^{[0]}(t)=\cP_{jk}(t)$, see (\ref{specdec0}).  For $g>0$ 
$\sigma(\cL_t^{[g]})\setminus\{0\}\subset \{z | \Re z<  0\}$. Moreover
\be\label{2pea}
\cP_{0}^{[g]}(t)=\sum_{j=0}^{d-1}\cP_{j0}^{[g]}(t).
\ee
Accordingly, for $0<g<g_0$ fixed,  we introduce $\cV(t,s)_{0\leq s\leq t\leq 1}$  by
\be\label{decompvwp}
\cV(t,s)=\cW(t,0)\Psi_\eps(t,s)\cW^{-1}(s,0),
\ee
in keeping with  Lemma \ref{decvwp}, where $\cW(t,s)$ is the solution to
\bea\label{wkg}
&&\partial_t\cW(t,s)=\Big\{\cP_{00}^{[g]}{}'(t)\cP_{00}^{[g]}(t)+\sum_{1\leq j\leq d-1}\cP_{j0}^{[g]}{}'(t)\cP_{j0}^{[g]}(t)+\sum_{1\leq j\neq k\leq d}\cP_{jk}^{[g]}{}'(t)\cP_{jk}^{[g]}(t) \Big\}\cW(t,s),\nonumber\\
&&\phantom{xxxxxxx}:= \cK^{[g]}_t\cW(t,s),\nonumber\\
&&\cW(s,s)=\un, \ \ 0\leq s,t\leq 1,
\eea
which satisfies the intertwining property with respect to all spectral projectors of $\cL_t^{[g]}$ by construction, and
\be
\Psi_\eps(t,s)=\cP_{00}^{[g]}(0)+\sum_{1\leq j\leq d-1}\cP_{j0}^{[g]}(0)\e^{\int_s^t\lambda_j^{[g]}(u)du/\eps}+\sum_{1\leq j\neq k\leq d}\cP_{jk}^{[g]}(0)\e^{\int_s^t\lambda_{jk}^{[g]}(u)du/\eps}.
\ee
Note that $\cW(t,s)$ is independent of $\eps$ and since its generator is analytic in  $g\in \C\setminus \{0\}$, for $|g|$ small enough,
we have
\be
\sup_{g_0>g> 0}\|\cW(t,s)\|\leq C_W,
\ee
uniformly in $0\leq s\leq t\leq 1$, as revealed by straightforward estimates of the Dyson series expansion of the solution to (\ref{wkg}). 
Also, $\Re \sigma(\cL_t^{[g]})\leq 0$ for all $t\in [0,1]$, implies that for all $0\leq s\leq t\leq 1$,
\be
\sup_{\eps>0, g_0>g> 0}\| \Psi_\eps(t,s) \|\leq C_\Psi,
\ee
uniformly in $0\leq s\leq t\leq 1$.
One checks that the operator $\cV(t,s)$ satisfies
\bea
\left\{
\begin{matrix}
\eps\partial_t \cV(t,s) =(\cL_t^{[g]}+\eps\cK^{[g]}_t)\cV(t,s),\\
\cV(s,s)=\un,\hfill
\end{matrix}
\right.
\eea
and, 
\be\label{vunifbdd}
\sup_{\eps>0, g > 0}\|\cV(t,s) \|_\tau\leq C_V.
\ee
uniformly in $0\leq s\leq t\leq 1$. Recall also $\|\cU(t,s)\|_\tau= 1$. 

As expected, $(\cV(t,s))_{0\leq s\leq t\leq 1}$ approximates the propagator $(\cU(t,s))_{0\leq s\leq t\leq 1}$ solution to (\ref{ulind}) in the slow drive regime $\eps\ra 0$, $g\ra 0$, $\eps \ll g$:
\begin{prop}\label{je<g}Assume {\bf Reg}, {\bf Gen} and {\bf Split}. Then, there exists $g_0>0$ and $C<\infty$ such that for $g < g_0$,  and all $0\leq s\leq t\leq 1$
\be
\|\cU(t,s)-\cV(t,s)\|_\tau\leq C \eps/g.
\ee
\end{prop}
\begin{rem}
i) In case some eigenvalues are permanently degenerate in (\ref{specdeclg}), the same result holds, {\it mutatis mutandis}. This is the case for  dephasing Lindbladians, as proven in \cite{AFGG} for $g$ fixed; see also \cite{J2} for results along these lines in an analytic context.\\
ii) The condition $\eps\ll g$ to get a useful approximation stems from the operators $\cR_j(B)$ and their derivatives in the integration by parts formula that contain differences of eigenvalues at the denominators, see (\ref{rbdis}), and hence have norms of order $1/g$.\\
iii) The approximation $\cV(t,s)$ depends on both $g$ and $\eps$ and requires the spectral data of $\cL_t^{[g]}$.
\end{rem}
\proof
This is a direct application of the integration by parts argument presented in Appendix, Lemma \ref{idint}, keeping track of the dependence in $g>0$ of the estimates; one makes use  of  the regularity of the spectral data of $\cL_t^{[g]}$ as $g\ra 0^+$ proven in Proposition \ref{regperpro}, and of the fact that both $\cU(t,s)$ and $\cV(t,s)$ are uniformly bounded in $\eps$ and $g$. 
More precisely, dropping the arguments in the notation, in the expression provided by Lemma \ref{idint} for $\cU-\cV$ with $\cX=\cU$,  $\cY=\cV$,  $\cG=\cL^{[g]}$, and $\cK=\cK^{[g]}$, 
the operators $\cX$  and $\cY$ are uniformly bounded in $\eps$ and $g$. The other operators appearing in (\ref{formulibp}) only depend on $g$ and involve $\cL^{[g]}$, $\cK^{[g]}$ and operators of the form 
$\cR_j(\cB)$ and $\partial_t\cR_j(\cB)$, where, see Remark \ref{rem81},
\be\label{41}
\cR_j(\cB)=\sum_{ k\neq j} \frac{\cP_j\cB\cP_k+\cP_k\cB\cP_j}{g_k-g_j}, 
\ee
with $\{g_j\}=\sigma(\cG)= \sigma(\cL^{[g]})$ with corresponding eigenprojectors $\cP_j$, and $\cB=[\cK^{[g]},\cP_j]$.  Note also
\be\label{42}
\partial_t\cR_j(\cB)=\cR_j(\partial_t\cB)+\sum_{ k\neq j} (\partial_t (g_j-g_k))\frac{\cP_j\cB\cP_k+\cP_k\cB\cP_j}{(g_k-g_j)^2}+ \frac{\cP_j'\cB\cP_k+\cP_k'\cB\cP_j+\cP_j\cB\cP_k'+\cP_k\cB\cP_j'}{g_k-g_j}.
\ee
By Proposition \ref{regperpro}, $\cL^{[g]}$, $\cK^{[g]}$, $\cP_j$ and their derivatives are uniformly bounded as $g\ra 0$, hence the same holds for $\cB$. Also, the denominators in (\ref{41}), (\ref{42}) never vanish for small $g>0$, and there exists $c>0$, such that  $\inf_{t\in [0,1], g>0}|g_j-g_k|\geq cg$ if both $g_j$ and $g_k$ stem from $\ker \cL^0$, while $\inf_{t\in [0,1], g>0}|g_j-g_k|\geq c$ otherwise. In the former case, 
$\sup_{t\in [0,1], g>0}|\partial_t(g_j-g_k)/(g_j-g_k)^2|\leq c/g$. 
 Altogether, this shows that $\max(\|\cR_j(\cB)\|, \|\partial_t \cR_j(\cB)\|)=\ode (1/g)$, which implies in turn $\|\cU-\cV\| =\ode (\eps/g).$
\qed

The next task  to get Theorem \ref{betcont} is to make more explicit the dependence in $(g,\eps)$ of the result above in the regime $g\gg \eps$. This goal can actually be achieved as a Corollary of another approximation of $\cU(t,s)$ derived in the next section under more general spectral assumptions, in the transition regime  $g \ll \sqrt \eps \ll 1$, see Corollaries \ref{slowdrive} and \ref{corevolad}.

\section{Transition Regime $g \ll \sqrt \eps \ll 1$}\label{trareg}

We work here under the assumptions  {\bf Reg} and {\bf Spec} with $\sigma_j(t)=\{e_j(t)\}$ for all $1\leq j\leq d$, all $t\in[0,1]$, so that $H(t)=\sum_{1\leq k\leq d}e_k(t)P_k(t)$, where $\dim P_k(t)\leq \infty$.\medskip

With 
$\lambda_{jk}(t)=-\i (e_j(t)-e_k(t))$, $\cP_{jk}(t)(\cdot)=P_{j}(t)\cdot P_k(t)$, see (\ref{specl0}) and (\ref{specpro}), and $ \cP_0(t)(\cdot)=\sum_{1\leq k\leq d}P_k(t)\cdot P_k(t)$ given by (\ref{projoker}), 
we denote by $\{\lambda_1(t), \lambda_2(t), \dots, \lambda_r(t)\}$ the non-zero distinct values in $\{\lambda_{jk}(t)\}_{1\leq j, k\leq d}$, where $2\leq r \leq d(d-1)$. Accordingly we define $r$ projectors on $\cB(\cH)$ by
\begin{align}
\cP_n(t)&=\hspace{-.2cm}\sum_{j\neq k \ {\rm s.t.} \atop  \lambda_{jk}(t)= \lambda_{n}(t)} \hspace{-.2cm}\cP_{jk}(t), \ \ \ \ 1\leq n\leq r.
\end{align}
Then, regardless of the dimension of the projectors $P_k(t)$, we have the smooth spectral decomposition
\be\label{specdeclog}
\cL_t^0=\sum_{1\leq n\leq r}\lambda_{n}(t) \cP_{n}(t)+ 0 \cP_0(t), \ \ \mbox{and} \ \ \sigma(\cL_t^0)=\{0\}\cup\{\lambda_{n}(t)\}_{1\leq n\leq r}.
\ee
The spectral projectors $\cP_n(t)$ have arbitrary dimension, possibly infinite. In particular, 
\be
\ker \cL_t^0=\Big\{A\in \cB(\cH)\ {\rm s.t.} \ A=\sum\nolimits_{1\leq j\leq d} P_j(t)AP_j(t)\Big\},
\ee
the set of diagonal operators with respect to $\{P_j(t)\}_{1\leq j\leq d}$.\medskip

Since we are interested in the transitions between spectral projectors of the Hamiltonian, that belong to the kernel of the Lindbladian at zero coupling, we focus on the evolution restricted to the projector $\cP_0^{[g]}(t)$ associated with the piece of spectrum of $\cL_t^{[g]}=\cL_t^0+g\cL_t^1$ a distance  of order $g$ away from zero, see (\ref{2pea}), by perturbation theory.

Let $\cW^{[g]}_0(t,s)$ be the Kato operator defined by
\bea\label{wgup}
\left\{\begin{matrix}
\partial_t \cW^{[g]}_0(t,s)=[\cP_0^{[g]}{}'(t),\cP_0^{[g]}(t)]\cW^{[g]}_0(t,s),\\
 \cW^{[g]}_0(s,s)=\un,\hfill
 \end{matrix}\right.
\eea
and $\Psi_\eps^{[g]}(t,s)$, $0\leq s\leq t\leq 1$, be the dynamical phase operator solution to
\bea\label{phigeps}
\left\{\begin{matrix}
\eps\partial_t \Psi^{[g]}_\eps(t,s)={\cW^{[g]}_0}(0,t)\cL_t^{[g]}\cW^{[g]}_0(t,0)\Psi^{[g]}_\eps(t,s),\\
 \Psi^{[g]}_\eps(s,s)=\un,\hfill
 \end{matrix}\right.
\eea
which commutes with $\cP_0^{[g]}(0)$.
Similarly to (\ref{decompvwp}), we set for $0\leq s\leq t\leq 1$,
\be\label{decadwp}
\cV_0^{[g]}(t,s)=\cW^{[g]}_0(t,0)\Psi^{[g]}_\eps(t,s)\cW^{[g]}_0(0,s).
\ee
The generators in these evolution equations being bounded on $\cT(\cH)$, the corresponding propagators have finite operator norms on $\cT(\cH)$ as well. We have the estimates
\begin{lem}\label{adiag0} Assume {\bf Reg} and {\bf Spec} with $\sigma_j(t)=\{e_j(t)\}$ for all $1\leq j\leq d$. 
There exist $C<\infty$, $C_{\Psi}<\infty$, $\eps_0>0$ and $g_0>0$ such that for all $\eps<\eps_0$ and $g<g_0$, 
\begin{align} \label{unifg}
&\|\cU(t,s)-\cV^{[g]}_0(t,s)\|_\tau \leq C\eps,\\
&\sup_{0\leq s\leq t\leq 1}\|\Psi^{[g]}_\eps(t,s)\|_\tau\leq C_{\Psi}, \label{estnormphi}
\end{align}
where 
\be
\cV^{[g]}_0(t,0)\cP_0^{[g]}(0)=\cP_0^{[g]}(t)\cV^{[g]}_0(t,0).
\ee
\end{lem}
\begin{proof}
The integration by parts argument in Appendix,  more precisely 
Corollary \ref{bddsmall} with $\cX=\cU$, $\cG=\cL^{[g]}$ and $\cY=\cV_0$, $\cK= [\cP_0^{[g]}{}',\cP_0^{[g]}]$, yield (\ref{unifg}) since
\be
\eps\partial_t \cV^{[g]}_0(t,s)=(\cL_t^{[g]}+\eps [\cP_0^{[g]}{}'(t),\cP_0^{[g]}(t)])\cV^{[g]}_0(t,s),
\ee
 and $\|\cU(t,0)\|_\tau=1$.
The uniformity  in $g>0$ of the estimate   is ensured by perturbation theory: the fact that $\cP_0^{[g]}(t)$ is associated with a piece of the spectrum of $\cL_t^{[g]}=\cL_t^0+g\cL_t^1$, of size of order $g$ and  separated by a gap of order $1$ from the rest of the spectrum, implies that for all $t\in [0,1]$
\be\label{perpo}
\cP_0^{[g]}(t)=\cP_0(t)+\ode(g).
\ee
This estimate remains true for derivatives with respect to $t$, so that
\be
[\cP_0^{[g]}{}'(t),\cP_0^{[g]}(t)]=[\cP_0'(t),\cP_0(t)]+\ode(g),
\ee
which yields uniformity in $g$ of the estimate so that $\sup_{0\leq s\leq t\leq 1 \atop 0 <\eps<\eps_0, 0< g<g_0}\|\cV^{[g]}_0(t,s)\|_\tau<\infty$. Moreover, uniformly in $s,t\in [0,1]$,
\be\label{pertwg}
\cW^{[g]}_0(t,s)=\cW_0(t,s)+\ode(g),
\ee
where $\cW_0(t,s)$ is defined by (\ref{w0}), as a consequence of Duhamel formula
\be
\cW^{[g]}_0(t,s)=\cW_0(t,s)+\int_{s}^t \cW_0(t,u)([\cP_0^{[g]}{}'(u),\cP_0^{[g]}(u)]-[\cP_0'(u),\cP_0(u)])\cW^{[g]}_0(u,s)du.
\ee
The estimate (\ref{pertwg}) and (\ref{decadwp}), together with $\cW_0(t,s)^{-1}=\cW_0(s,t)$, see (\ref{proprel}), imply (\ref{estnormphi}). \qed
\end{proof}

\subsection{Reduced dynamics}
We now consider the reduced dynamics within $\cP_0(0)\cB(\cH)$, depending on a time scale $1/\delta$ and driven by the splitting operator, that will approximate $\Psi^{[g]}_\eps(t,0)\cP_0^{[g]}(0)$ in certain regimes.\\

Let $\tilde \cL_t^1=\cP_0(t) \cL_t^1 \cP_0(t)$ and recall, see (\ref{auxad}), that $\tilde \Psi_{\delta}(t,s)$, $0\leq s\leq t\leq 1$ is defined for $\delta>0$ by
\bea\label{auxad2}
\left\{\begin{matrix}
\delta \partial_t \tilde \Psi_{\delta}(t,s)={\cW_0}(0,t)\tilde \cL_t^1\cW_0(t,0)\tilde \Psi_{\delta}(t,s),\\
\tilde \Psi_{\delta}(s,s)=\un.\hfill 
\end{matrix}\right.
\eea
Note that the Dyson series for  $\tilde \Psi_{\delta}(t,0)$ has the same integral terms as those provided in Lemma \ref{gendiag}. Also by definition, { recall $\cQ_0(t)=\un -\cP_0(t)$,}
\be\label{cpartpsi}
\tilde \Psi_{\delta}(t,s)=\cP_0(0)\tilde \Psi_{\delta}(t,s)\cP_0(0)+\cQ_0(0),
\ee
where $\|\cQ_0(0)\|_\tau\leq 2$, since $\cP_0(0)$ is CPTP.

\medskip

The next Proposition is the main technical step regarding the approximation of the evolution $\cU(t,s)$ in the transition regime considered, which holds regardless of the dimension of the projectors $P_j(t)$.
\begin{prop}\label{tech} Assume {\bf Reg} and {\bf Spec} with $\sigma_j(t)=\{e_j(t)\}$, for all $1\leq j\leq d$.
There exist $\tilde C$, $\tilde C_\Psi<\infty$ and $g_0$, $\eps_0$, $\alpha_0>0$ such that for all $g\leq g_0$, $\eps\leq \eps_0$,  $g^2/\eps \leq \alpha_0$, and $0\leq s\leq  t\leq 1$,
\be
\|\Psi^{[g]}_\eps(t,s)\cP_0^{[g]}(0)-\tilde \Psi_{\eps/g}(t,s)\cP_0^{[g]}(0)\|_\tau\leq \tilde C (t-s)g^2/\eps,
\ee
and $\|\tilde \Psi_{\eps/g}(t,s)\|_\tau\leq \tilde C_\Psi$.
\end{prop}
\begin{rem}
The ratio $\eps/g$ which determines the time scale in the reduced dynamics $\tilde \Psi_{\eps/g}(t,s)$ is not required to be small here.
\end{rem}
\proof
Recall (\ref{deftl1}) which states that, uniformly in $0\leq t\leq 1$,
\begin{align} \label{recall}\cL_t^{[g]}\cP_0^{[g]}(t) = \cP_0^{[g]}(t)\cL_t^{[g]}\cP_0^{[g]}(t) 
= g\cP_0(t) \cL_t^1 \cP_0(t) +  \ode(g^2),
\end{align}
and the intertwining relation $\cW^{[g]}_0(t,s)\cP_0^{[g]}(s)=\cP_0^{[g]}(t)\cW^{[g]}_0(t,s)$, consequence of the definition (\ref{wgup}). Composing (\ref{recall}) by 
$\cW^{[g]}_0(0,t)$ and $\cW^{[g]}_0(t,0)$ and using 
 (\ref{perpo}) and (\ref{pertwg}), we get that
 the generator of $\Psi^{[g]}_\eps(t,s)\cP_0^{[g]}(0)$, see (\ref{phigeps}),  satisfies
\be
\cW^{[g]}_0(0,t)\cL_t^{[g]}\cW^{[g]}_0(t,0)\cP_0^{[g]}(0)-g\cW_0(0,t)\cP_0(t)\cL_t^1\cP_0(t)\cW_0(t,0)\cP_0^{[g]}(0)=g^2\Lambda(t,g), 
\ee
where $\|\Lambda(t,g)\|_\tau\leq C_\Lambda$, uniformly in $0\leq t\leq 1$ and $g>0$ small enough.
Therefore, making use of $[\Psi^{[g]}_\eps(t,s),\cP_0^{[g]}(0)]\equiv 0$, Duhamel formula yields
\begin{align}
\Psi^{[g]}_\eps(t,s)\cP_0^{[g]}(0)&=\tilde \Psi_{\eps/g}(t,s)\cP_0^{[g]}(0)+\frac{g^2}{\eps}\int_s^t\tilde \Psi_{\eps/g}(t,r)\Lambda(r,g)\Psi^{[g]}_\eps(r,s)\cP_0^{[g]}(0)ds.
\end{align}
Hence
\be\label{phiphi}
\| \Psi^{[g]}_\eps(t,s)\cP_0^{[g]}(0)-\tilde \Psi_{\eps/g}(t,s)\cP_0^{[g]}(0)\|_\tau\leq C_\Lambda (t-s) \frac{g^2}{\eps} \sup_{0\leq r\leq t\leq 1}\|\tilde \Psi_{\eps/g}(t,r)\|_\tau \sup_{0\leq r\leq s\leq 1}\|\Psi^{[g]}_\eps(r,s)\cP_0^{[g]}(0)\|_\tau.
\ee
Now, $\sup_{0\leq s\leq t\leq 1}\| \Psi^{[g]}_\eps(t,s)\cP_0^{[g]}(0)\|_\tau:=C_{\Psi P}$ is uniformly bounded for $\eps>0$ and $g>0$ small enough, see (\ref{estnormphi}). Moreover, thanks to (\ref{cpartpsi})  and (\ref{perpo}), there exists $c<\infty$ such that
\begin{align}
&\| \tilde \Psi_{\eps/g}(t,s)\|_\tau\leq \|\tilde \Psi_{\eps/g}(t,s)\cP_0(0)\|_\tau +2,\\ 
&\| \tilde \Psi_{\eps/g}(t,s)\cP_0^{[g]}(0)-\tilde \Psi_{\eps/g}(t,s)\cP_0(0)\|_\tau\leq c g\|\tilde \Psi_{\eps/g}(t,s)\|_\tau.
\end{align}
Consequently, { making use of the identity (recall (\ref{cpartpsi}))
\begin{align}
 \tilde \Psi_{\eps/g}(t,s)&=\cQ_0(0)+\Psi_{\eps}^{[g]}(t,s) \cP_0^{[g]}(0) \nonumber \\
 &+(\tilde \Psi_{\eps/g}(t,s)\cP_0^{[g]}(0)- \Psi^{[g]}_\eps(t,s)\cP_0^{[g]}(0))+ \tilde \Psi_{\eps/g}(t,s)( \cP_0(0)-\cP_0^{[g]}(0)).
\end{align}
we get with the above and} (\ref{phiphi}), 
\be
\sup_{0\leq s\leq t\leq 1}\|\tilde \Psi_{\eps/g}(t,s)\|_\tau\leq 2+ C_{\Psi P}\Big( 1+  C_\Lambda \frac{g^2}{\eps} \sup_{0\leq r\leq t\leq 1}\|\tilde \Psi_{\eps/g}(t,r)\|_\tau \Big)+cg \sup_{0\leq s\leq t\leq 1}\|\tilde \Psi_{\eps/g}(t,s)\|_\tau.
\ee
Therefore, there exists $0<\tilde C_\Psi<\infty$ such that if $\eps>0$, $g>0$ and $g^2/\eps$ are small enough, 
\be
\sup_{0\leq s\leq t\leq 1}\|\tilde \Psi_{\eps/g}(t,s)\|_\tau\leq\frac{2+C_{\Psi P}}{1-C_{\Psi P}C_\Lambda \frac{g^2}{\eps}-cg }\leq \tilde C_\Psi,
\ee
irrespectively of the value of $\eps/g$. By inserting this estimate into  (\ref{phiphi}), we get the first statement with  $\tilde C=C_\Lambda \tilde C_\Psi C_{\Psi P}$.
\qed\\

We are now in a position to approximate the evolution $\cU(t,s)$ and the transition probabilities between the spectral projectors $P_j(t)$ within $\ker \cL_t^0$, which is the content of Theorem \ref{thmtra}:
\begin{cor} \label{cor1} Assume  {\bf Reg} and {\bf Spec} with $\sigma_j(t)=\{e_j(t)\}$, for all $1\leq j\leq d$.  There exists $C_0<\infty$, $\eps_0>0$, $g_0>0$, and $\alpha_0>0$ such that for all $0\leq t\leq 1$, $\eps<\eps_0$, $g\leq g_0$, $g^2/\eps<\alpha_0$,
\be\label{appUtrag}
\|\cU(t,0)\cP_0(0)-\cW_0(t,0)\tilde \Psi_{\eps/g}(t,0)\cP_0(0)\|_\tau\leq C_0(\eps+g+g^2/\eps).
\ee
Consequently, for any state $\rho_j=P_j(0)\rho_j P_j(0)\in\cT(\cH)$,  for any $1\leq j,k\leq d$, 
\begin{align}
\tr \{P_k(t)\cU(t,0)(\rho_j)\}
&=\tr \{P_k(0) \tilde \Psi_{\eps/g}(t,0)(\rho_j)\}+\ode(\eps+g+g^2/\eps).
\end{align}
Also, for any $\delta>0$ and any $0\leq t\leq 1$, the map $\tilde \Psi_\delta(t,0)\cP_0(0)$ is CPTP on $\cT(\cH)$.
\end{cor}
\proof
The first statement follows immediately from Lemma \ref{adiag0}, Proposition \ref{tech} and estimate (\ref{perpo}). \\
To access the transition probabilities $\tr (P_k(t)\cU(t,0)(\rho_j))$, we get the action of $\cW_0(t,0)$ on the projector $\cP_0(0)$ by means of Lemma \ref{lemw}.
Thus, using $\tilde \Psi_{\eps/g}(t,0)\cP_0(0)=\cP_0(0)\tilde \Psi_{\eps/g}(t,0)\cP_0(0)$, we get  the transition probabilities in terms of the reduced dynamics in the regime $g^2\ll \eps\ll 1$
\begin{align}
\tr \{P_k(t)\cU(t,0)(P_j(0))\}&=\tr \{P_k(t) \cW_0(t,0)\circ \tilde \Psi_{\eps/g}(t,0)(P_j(0))\}+\ode(\eps+g+g^2/\eps)\nonumber\\
&=\tr \{P_k(0) \tilde \Psi_{\eps/g}(t,0)(P_j(0))\}+\ode(\eps+g+g^2/\eps).
\end{align}
Finally,  given $\delta>0$,  (\ref{appUtrag}) for $\eps=\delta g$ yields $\cW_0(t,0)\tilde \Psi_{\delta}(t,0)\cP_0(0)=\lim_{g\ra 0} \cU(t,0)\cP_0(0)$ in $\| \cdot \|_\tau$-norm , where 
$\cU(t,0)\cP_0(0)$ is CPTP on $\cT(\cH)$, so the same is true for $\cW_0(t,0)\tilde \Psi_{\delta}(t,0)\cP_0(0)$. Since $\cW_0(0,t)\cP_0(t)$ is CPTP as well, see Lemma \ref{wpcptp}, and $\tilde \Psi_{\delta}(t,0)\cP_0(0)=\cP_0(0)\tilde \Psi_{\delta}(t,0)\cP_0(0)$, 
\be
\tilde \Psi_{\delta}(t,0)\cP_0(0)=\cW_0(0,t)\cW_0(t,0)\cP_0(0)\tilde \Psi_{\delta}(t,0)\cP_0(0)=\cW_0(0,t)\cP_0(t)\cW_0(t,0)\tilde \Psi_{\delta}(t,0)\cP_0(0)
\ee
is CPTP, as a composition of such maps.
\qed
\subsection{Associated Markov Process}\label{secmarkov}
Let us proceed with a remark about the generic finite dimensional case. If {\bf Reg} and {\bf Gen} hold (without condition on the Bohr frequencies, actually),  the generator $\cW_0(0,t)\tilde\cL_t^1\cW_0(t,0)$ of the reduced dynamics $\tilde \Psi_{\delta}(t,0)$ has a matrix expression in the fixed basis $\{P_1(0), P_2(0), \dots, P_d(0)\}$ given by the time dependent matrix $\tilde L(t)$  (\ref{mattildel}). Indeed, (\ref{L1M1}) and Remark \ref{remwp} i), yield
\be
\cW_0(0,t)\tilde\cL_t^1\cW_0(t,0)(P_j(0))=\cW_0(0,t)\cP_0(t)\cL_t^1(P_j(t))=\sum_{1\leq k\leq d}P_k(0)\tilde L_{kj}(t).
\ee
In other words, $\cW_0(0,t)\tilde\cL_t^1\cW_0(t,0)\simeq \tilde L(t)$, where $\sum_{1\leq k\leq d}\tilde L_{kj}(t)=0$ for any $1\leq j\leq d$, recall (\ref{trastoch}). Hence, the matrix representation of the reduced dynamics in the same basis,  $\tilde \Psi_{\delta}(t,0)|_{\spa \{P_1(0),\dots, P_d(0)\}}$, is such that its transpose  is a stochastic matrix, see {\it e.g.} \cite{YZ}. 
Therefore, we can associate  to the reduced dynamics a $d$-state classical continuous-time Markov process:
\begin{lem}\label{markov} Under {\bf Reg} and {\bf Gen}, the reduced dynamics $\tilde \Psi_{\delta}(t,0)\cP_0(0)$ is associated to a continuous-time Markov process $(X_t)_{t\geq 0}$ in the state space $\{P_1(0), \dots, P_d(0)\} := \{1, \dots, d\}$ by the relation for any $t\geq 0$
\be
\P(X_t=j | X_0=i)=\tr \big(P_j(0)\tilde \Psi_{\delta}(t,0)(P_i(0))\big).
\ee
\end{lem}

\subsection{Back to the slow drive regime}

Specialising to the simpler generic  framework given by assumptions {\bf Gen}, and supposing the dissipator splits $\ker \cL_t^0$ maximally, we can further approximate the reduced evolution $\tilde \Psi_{\delta}(t,s)$ for $\delta=\eps/g\ll 1$. \\

By Assumption {\bf Split}, Remark \ref{gershg} and (\ref{sectilde}), the generator  of $\tilde \Psi_{\delta}(t,0)$ reads
\be\label{effad}
{\cW_0}(0,t)\tilde \cL^1_t{\cW_0}(t,0)=\cP_0(0)\Big(0 \tilde \cQ_0(t)+\sum_{j=1}^{d-1}\tilde \lambda_j(t) \tilde \cQ_j(t)\Big)\cP_0(0) + 0\cQ_0(0),
\ee
with rank one spectral projectors $\tilde \cQ_j(t)={\cW_0}(0,t)\tilde \cP_j(t){\cW_0}(t,0)$ and corresponding eigenvalues $\tilde \lambda_j(t)$ with negative or zero real parts.  
Thus $\tilde \Psi_{\delta}(t,s)\cQ_0(0)\equiv \cQ_0(0)$, and $[\tilde \Psi_{\delta}(t,s),  \cP_0(0)]\equiv 0$. 
In case the time scale $1/\delta$ is large,  the following adiabatic approximation holds: There exists $\delta_0>0$ and $\tilde c<\infty $ such that for all $\delta<\delta_0$, and all  
$0\leq s \leq t\leq 1$, 
\be\label{adphiaux}
\Big\|\tilde \Psi_{\delta}(t,s)\cP_0(0)-\tilde \cW(t,0)\Big(\sum_{j=0}^{d-1}\e^{\int_s^t\tilde \lambda_j(r)dr/\delta} \tilde \cQ_j(0)\Big)\tilde \cW(0,s)\cP_0(0)\Big\| \leq \tilde c \delta,
\ee
where $\tilde \cW(t,s)$ is defined by
\begin{align}\label{katilde}
\left\{\begin{matrix}
 \partial_t   \tilde \cW(t,s) = \Big( \sum_{j=0}^{d-1} \tilde \cQ_j'(t)\tilde \cQ_j(t)  \Big) \tilde \cW(t,s)  \\
 \tilde \cW(s,s)=\un, \hfill
\end{matrix} \right.
\end{align}
so that $\tilde \cW(t,s)\cQ_0(0)=\cQ_0(0)\tilde \cW(t,s)\equiv \cQ_0(0)$ and the following non trivial intertwining relations  hold for $0\leq j\leq d-1$
\be\label{intertilde}
\tilde \cW(t,s)\tilde \cQ_j(s)=\tilde \cQ_j(t)\tilde \cW(t,s).
\ee
Indeed, the integration by parts argument Lemma \ref{idint} applies with $\eps=\delta$ to $\cX(t,s)=\tilde\Psi_\delta(t,s)\cP_0(0)$, $\cY(t,s)=\tilde \cW(t,0)\Big(\sum_{j=0}^{d-1}\e^{\int_s^t\tilde \lambda_j(r)dr/\delta} \tilde \cQ_j(0)\Big)\tilde \cW(0,s)\cP_0(0)$, $\cG(t)={\cW_0}(0,t)\tilde \cL^1_t{\cW_0}(t,0)\cP_0(0)$,  and $\cK(t)=\sum_{j=0}^{d-1} \tilde \cQ_j'(t)\tilde \cQ_j(t) $. Since $\cY(t,s)$ is uniformly bounded in $\delta$, Corollary \ref{bddsmall} yields estimate (\ref{adphiaux}).

Hence, in the restricted slow drive regime $\eps\ll g\ll \sqrt \eps\ll 1$,  we can take advantage of (\ref{adphiaux}) to express $\tilde \Psi_{\eps/g}(t,0)$ in terms of the spectral data of $\tilde \cL_t^1$, in the framework given by assumption {\bf Gen} { to approximate $\cU(t,0)\cP_0(0)$.  Indeed, making use of (\ref{appUtrag}), (\ref{adphiaux}) for $s=0$, and $\cP_0(0) \tilde \cQ_j(0)= \tilde \cQ_j(0)$, for all $0\leq j\leq d-1$,  and taking into account the regime considered, we immediately get the} 
\begin{cor}\label{corevolad} Assume {\bf Reg},  {\bf Gen}  and {\bf Split}. 
There exist $C_1<\infty$, $\eps_0>0$, $g_0>0$, $\alpha_0>0$ and $\beta_0>0$ such that for all $0\leq t\leq 1$, $\eps\leq\eps_0$, $g\leq g_0$, $g^2/\eps\leq \alpha_0$, $\eps/g\leq \beta_0$
\be
\Big\|\cU(t,0)\cP_0(0)-\cW_0(t,0)\tilde \cW(t,0)\Big(\sum_{j=0}^{d-1}\e^{\frac{g}{\eps}\int_0^t\tilde \lambda_j(r)dr } \tilde \cQ_j(0)\Big)\Big\|\leq C_1(g^2/\eps+\eps/g).
\ee
\end{cor}
In order to compute the transition probability between the eigenprojectors of the Hamiltonian, we make explicit  $\tilde \cQ_j(t)$, $0\leq j\leq d-1$, the rank one eigenprojectors of ${\cW_0}(0,t)\tilde \cL^1_t{\cW_0}(t,0)$.  
In keeping with (\ref{BRAKET}), for $A,B\in \cB(\cH)$, we define a rank one operator on  $\cB(\cH)$ by
\be
\big| A\KET \BRA B\big| : C \mapsto \BRA B, C\KET A= \tr (B^* C)A.
\ee
Hence there exist $\nu_j(t), \mu_j(t) \in \cP_0(0)\cB(\cH)=\spa \{P_k(0), k\in\{1, \dots, d\}\}$ such that
\begin{align}\label{normqt}
&\tilde \cQ_j(t)=\big|\nu_j(t)\KET \BRA \mu_j(t)\big|, \  \ \mbox{where}  \ \ 
\nonumber \\ &
\BRA \mu_j(t), \nu_j(t)\KET\equiv 1, \ \ \ \BRA  \mu_j(t), \mu_j(t)\KET\equiv d.
\end{align}
where the last identity serves normalisation purposes. 
Since ${\cW_0}(0,t)\tilde \cL^1_t{\cW_0}(t,0)$ is smooth, these operators can be chosen smooth as well. 
Moreover, for $\tilde \cW(t,s)$ defined by (\ref{katilde}), we have for all $j\in\{0,1,\dots, d-1\}$
\be
\tilde \cW(t,s)(\nu_j(s))=\nu_j(t)\e^{-\int_s^t \small \langle \hspace{-.05cm} \small \langle \mu_j(u),\partial_u \nu_j(u) \small \rangle \hspace{-.05cm} \small \rangle du}.
\ee
This identity follows from (\ref{intertilde}) together with $\tilde \cQ_j(t) \partial_t \{\tilde \cW(t,s)(\nu_j(s))\}\equiv 0$. 

\medskip

In particular, for $j=0$, the eigenprojector $\cQ_0(t)$ associated with $\lambda_0(t)\equiv 0$ takes the form
\be
\tilde \cQ_0(t)=\big|\nu_0(t)\KET \BRA \un \big|,
\ee
where $\nu_0(t)=\cP_0(0)(\nu_0(t))\in \ker {\cW_0}(0,t)\tilde \cL^1_t{\cW_0}(t,0)$.  Moreover, (\ref{normqt}) implies
\be
\tilde \cW(t,s)(\nu_0(s))=\nu_0(t),
\ee
and the justification that $\mu_0(t)=\un$ stems from $\tr \cL_t^1(\rho)\equiv 0$, see (\ref{traprol1}).
Equivalently, by Lemma \ref{lemw}, $\nu_0(t)$ is characterised by 
\be\label{eqcharnu}
\tilde \cL^1_t(\tilde \nu_0(t))=0 \ \ \mbox{where} \ \ 
\tilde \nu_0(t)=\cW_0(t,0)(\nu_0(t))=W(t,0) \nu_0(t) W(0,t).
 \ee

Note that since $\Re \tilde \lambda_j(t)<0$ for $j\neq 0$, we get that  for all fixed $t>0$, $\e^{\frac{g}{\eps}\int_0^t\tilde \lambda_j(r)dr } =\ode ((\eps/g)^\infty)$. Hence, from Corollary \ref{corevolad} and (\ref{eqcharnu}), for all fixed $0<t<1$, and all $P_j(0)$, we have
\be
\cU(t,0)(P_j(0))=\tilde \nu_0(t)+\ode(g^2/\eps+\eps/g),
\ee
which is the statement of Theorem \ref{betcont}.

\medskip

Let us turn to the transition probabilities. Writing for any $1\leq l\leq d$
\be
\nu_l(t)=\sum_{1\leq k\leq d} \tr (P_k(0)\nu_l(t)) P_k(0):=\sum_{1\leq k\leq d} \nu_l^k(t) P_k(0),
\ee
we thus have
\be
\tr \{P_k(0)\tilde \cW(t,0)\circ \tilde \cQ_l(0)(P_j(0))\}=\bar \mu_l^j(0) \nu_l^k(t)\e^{-\int_0^t \small \langle \hspace{-.05cm} \small \langle \mu_l(s),\nu_l'(s) \small \rangle \hspace{-.05cm} \small \rangle ds}.
\ee
We are in a position to estimate the transition probabilities in the adiabatic regime for the reduced evolution, to complete Theorem \ref{betcont}:
\begin{cor}\label{slowdrive}
Assume {\bf Reg}, {\bf Gen}  and {\bf Split}.  Then  $\forall t\in[0,1]$, $j\neq k$, we have for $\eps\ll g\ll \sqrt \eps\ll 1$,
\be
\tr (P_k(t) \cU(t,0) (P_j(0)))=\sum_{0\leq l\leq d-1}\e^{g/\eps \int_0^t\tilde \lambda_j(s)ds}\,\bar \mu_l^j(0) \nu_l^k(t)\e^{-\int_0^t \small \langle \hspace{-.05cm} \small \langle \mu_l(s),\nu_l'(s) \small \rangle \hspace{-.05cm} \small \rangle ds}+\ode(g^2/\eps+\eps/g).
\ee
In particular, for any fixed $t>0$, we have in the same regime
\be
\tr (P_k(t) \cU(t,0) (P_j(0)))=\tr(P_k(t)\tilde \nu_0(t))+\ode(g^2/\eps+\eps/g),
\ee
where $\tilde \nu_0(t)$ is uniquely defined by $\tilde \cL_t^1(\tilde \nu_0(t))=0$ and $\tr  (\tilde \nu_0(t))=1$.
\end{cor}
\begin{rem}
The first statement stems from Corollary \ref{corevolad}, while the second one takes advantage of $\e^{g/\eps \int_0^t\tilde \lambda_j(s)ds}=~\ode((\eps/g)^\infty)$ for $j>0$ if $t>0$ is independent of $\eps/g$, since $\Re \tilde \lambda_j(s)<0$ for such $j$'s. The reformulation of the leading order is a consequence of the considerations above and $\cW_0(t,s)\cP_0(s)$ being trace preserving.
\end{rem}

\subsection{Back to the perturbative regime}

Finally, we briefly check that Corollary \ref{cor1} reduces to a statement of Proposition \ref{propweak} in the perturbative regime $g\ll \eps$. We note that the definition (\ref{auxad2}) allows for an approach of $\tilde \Psi_{\delta}$ via Dyson series which gives for $\delta=\eps/g\gg 1$,
\be
\tilde \Psi_\delta(t,0)=\un +\frac{1}{\delta}\int_0^t \cW_0(0,s)\tilde \cL_s^1 \cW_0(s,0)ds +\ode(1/\delta^2),
\ee
since $\|\tilde\Psi_{\delta}(t,s)\|_\tau$ is uniformly bounded in $0\leq s\leq t\leq 1$ and $\delta$.  Hence, given the definition of $\tilde \cL_s^1$,
\begin{align}
\cW_0(t,0)\tilde \Psi_{\eps/g}(t,0)\cP_0(0)&=\cW_0(t,0)\cP_0(0)+\cW_0(t,0)\frac{g}{\eps}\int_0^t \cW_0(0,s)\tilde \cL_s^1 \cW_0(s,0)\cP_0(0)ds+\ode(g^2/\eps^2)\nonumber\\
&=\cW_0(t,0)\cP_0(0)+\frac{g}{\eps}\cW_0(t,0)\int_0^t \cP_0(0)\cW_0(0,s) \cL_s^1 \cW_0(s,0)\cP_0(0)ds+\ode(g^2/\eps^2).
\end{align}
Since $\sigma_j(t)=\{e_j(t)\}$ for all $1\leq j\leq d$ and all $t\in [0,1]$, Remark \ref{v=w} {   applies which, noting the intertwining relation (\ref{intertw0}), yields}
 for any $0\leq s\leq t\leq 1$,
\be\label{vpwf}
\cW_0(t,s)\cP_0(s)=\cV^0(t,s)\cP_0(s),
\ee
 where $\cV^0(t,s)$ is defined in (\ref{defv00}).
 Thus, further assuming $g\ll \eps$ in Corollary \ref{cor1}, we recover the perturbative regime estimate of Proposition (\ref{propweak}) for $N=1$ under the form
\begin{cor}\label{recovpereg} Assume  {\bf Reg} and {\bf Spec} with $\sigma_j(t)=\{e_j(t)\}$, for all $1\leq j\leq d$ and all $t\in [0,1]$. In the regime $(\eps, g)\ra (0,0)$ and $g/\eps\ra 0$, (\ref{appUtrag}) yields
\be
\cU(t,0)\cP_0(0)=\cV^0(t,0)\cP_0(0)+\frac{g}{\eps}\int_0^t \cP_0(t)\cV^0(t,s) \cL_s^1 \cV^0(s,0)\cP_0(0)ds+\ode(\eps+(g/\eps)^2),
\ee
where $\cV^0(t,s)\cP_0(s)=\cP_0(t)\cV^0(t,s)=\cP_0(t)\cW_0(t,s)\cP_0(s)$ is independent of $\eps$.
\end{cor}
\begin{rem} i) The error term is smaller than the explicit integral term for $\eps^2\ll g \ll \eps$.\\
ii) We recover this way the second statement of Theorem \ref{asuper}.
\end{rem}

\section{Example}\label{ex}

We consider here a two-level system or Qubit, for which the reduced dynamics $\tilde \Psi_\delta(t,s)$ can be computed explicitly under some symmetry of the Lindblad operators. Beyond its intrinsic interest, this example allows us to illustrate the different regimes we encountered in the general case.
\medskip

We assume {\bf Reg} and {\bf Gen} and with the notations introduced so far, for $\cH=\C^2$, we consider the two-level Hamiltonian 
\be
H(t)=\sum_{1\leq j\leq 2} e_j(t)P_j(t),\ \  \mbox{with} \ \ P_j(t)=|\ffi_j(t)\ket\bra\ffi_j(t)|,
\ee 
 and $\ffi_j(t)=W(t,0)\ffi_j(0)$, $W(t,0)$ being the unitary Kato operator.
We assume the dissipator
\be\label{diss2}
\cL_t^1(\cdot)=\sum_{l\in I}\Big(\Gamma_l(t)\cdot \Gamma_l^*(t)-\frac12\big\{\Gamma_l^*(t)\Gamma_l(t), \cdot\big\}\Big),
\ee
has jump operators in $\cB(\C^2)$ satisfying the symmetry condition
\be\label{symlindop}
\sum_{l\in I}|\bra \ffi_1(t)|\Gamma_l(t)\ffi_2(t)\ket|^2=\sum_{l\in I}|\bra \ffi_2(t)|\Gamma_l(t)\ffi_1(t)\ket|^2.
\ee
This is the case in particular if all jump operators are self-adjoint. Note that condition (\ref{symlindop}) is independent of the normalised basis of eigenvectors of $H(t)$ used to express it.

Then we have:
\begin{prop}
Let $\cH=\C^2$, and assume {\bf Reg} and {\bf Gen} for $d=2$. Further suppose the dissipator satisfies the symmetry condition (\ref{symlindop}) and set $\gamma(t)=\sum_{l\in I}|\bra \ffi_1(t)|\Gamma_l(t)\ffi_2(t)\ket|^2\in \R^+$. Then, for any $\delta>0$, the reduced dynamics takes the explicit form in the ordered basis $\{P_1(0), P_2(0)\}$ of $\cP_0(0)$:
\begin{align}
\tilde\Psi_\delta(t,s)|_{\spa \{P_1(0), P_2(0)\}}
=\frac12\begin{pmatrix}
1 &1  \cr 
1  & 1 \cr
\end{pmatrix}+\e^{-\frac2\delta\int_s^t\gamma(u)du} \frac12\begin{pmatrix}
1 &-1  \cr 
-1  & 1 \cr
\end{pmatrix}.
\end{align}
Hence, in the regime $g\ll \sqrt \eps \ll 1$,  for any initial state $\rho_0=\rho_1(0)P_1(0)+\rho_2(0)P_2(0)$, any $t\in [0,1],$
\begin{align}\label{qubitapp}
\cU(t,0)(\rho_0)=\tilde \rho_1(t)P_1(t)+\tilde \rho_2(t)P_2(t)+\ode(\eps+g+g^2/\eps),
\end{align}
where 
\begin{align}
\tilde \rho_1(t)&=\frac12\Big(1+\e^{-2\frac{g}{\eps}\int_0^t\gamma(s)ds}(\rho_1(0)-\rho_2(0))\Big),\nonumber\\
\tilde \rho_2(t)&=\frac12\Big(1+\e^{-2\frac{g}{\eps}\int_0^t\gamma(s)ds}(\rho_2(0)-\rho_1(0))\Big).
\end{align}
In particular, the transition probabilities read in the same regime
\begin{align}\label{qubtrapro}
\tr (P_2(t)\cU(t,0)(P_1(0)))=\tr (P_1(t)\cU(t,0)(P_2(0)))=\frac12\Big(1-\e^{-2\frac{g}{\eps}\int_0^t\gamma(s)ds}\Big)+\ode(\eps+g+g^2/\eps).
 \end{align}
\end{prop}
\begin{rem}
i) For $0<t \leq1$ fixed such that $\int_0^t\gamma(s)ds>0$, if $\eps \ll g\ll \sqrt \eps \ll 1$, 
\be
\cU(t,0)(\rho_0)=\frac12\un+\ode(g^2/\eps+(\eps/g)^\infty),
\ee
which corresponds to Theorem \ref{betcont}. Note that $\un$ spans $\ker \tilde \cL_t^1|_{\cP_0(t)\cH}$.\\
ii) If $ g\ll \eps \ll \sqrt \eps \ll 1$, 
\begin{align}
\cU(t,0)(\rho_0)=&\rho_1(0)P_1(t)+\rho_2(0)P_2(t)\nonumber\\
 &-\frac{g}{\eps}\int_0^t\gamma(s)ds\,  (\rho_1(0)-\rho_2(0)) (P_1(t) - P_2(t))
+\ode(\eps+g^2/\eps^2)
\end{align}
which corresponds to Corollary \ref{recovpereg} and Theorem \ref{asuper}.\\
iii) The  state $\cU(t,0)(\rho_0)$ is determined by the asymptotics of the scalar factor $\e^{-2\frac{g}{\eps}\int_0^t\gamma(s)ds}$. In case $g\simeq \eps$, {\it i.e.} $g=\alpha \eps$, for some fixed $\alpha>0$, the leading order of $\cU(t,0)(\rho_0)$ takes the form of any diagonal state with respect to the eigenbasis of $H(t)$, depending on the value  of $\e^{-2\alpha\int_0^t\gamma(s)ds}$.\\
iv) The Markov process interpretation of Lemma \ref{markov} remains in force here, with $\tilde\Psi_\delta(t,0)|_{\spa \{\cP_0(0)\cB(\C^2)\}}$ being bistochastic.
\end{rem}
\begin{proof}
The arguments leading to Lemma \ref{markov} show that the generator of the reduced dynamics $\cW_0(0,t)\tilde\cL_t^1\cW_0(t,0)$ has the following matrix form in the basis $\{P_1(0), P_2(0)\}$
(dropping the variable $t$ from the notation)
\begin{align}
\tilde L=&\sum_{l\in I}
\begin{pmatrix}
|\bra\ffi_1|\Gamma_l\ffi_1\ket|^2 - \|\Gamma_l\ffi_1\|^2 & |\bra\ffi_1|\Gamma_l\ffi_2\ket|^2  \cr 
|\bra\ffi_2|\Gamma_l\ffi_1\ket|^2  & |\bra\ffi_2|\Gamma_l\ffi_2\ket|^2 - \|\Gamma_l\ffi_2\|^2 \cr
\end{pmatrix}=\sum_{l\in I}
\begin{pmatrix}
-|\bra\ffi_2|\Gamma_l\ffi_1\ket|^2  & |\bra\ffi_1|\Gamma_l\ffi_2\ket|^2  \cr 
|\bra\ffi_2|\Gamma_l\ffi_1\ket|^2  & -|\bra\ffi_1|\Gamma_l\ffi_2\ket|^2 \cr
\end{pmatrix},
\end{align}
thanks to property (\ref{trastoch}). The assumed symmetry  (\ref{symlindop}) allows us to get (restoring the time variable)
\begin{align}
\tilde L(t)=&
\sum_{l\in I}|\bra\ffi_2(t)|\Gamma_l(t)\ffi_1(t)\ket|^2
\begin{pmatrix}
-1 &1  \cr 
1  & -1 \cr
\end{pmatrix}= \gamma(t) \begin{pmatrix}
-1 &1  \cr 
1  & -1 \cr
\end{pmatrix}.
\end{align}
Therefore, in the same basis, the reduced dynamics
solution  to (\ref{auxad2}) with $\delta>0$ reads
\begin{align}
\tilde\Psi_\delta(t,s)|_{\spa \{P_1(0), P_2(0)\}}=\frac12\begin{pmatrix}
1 &1  \cr 
1  & 1 \cr
\end{pmatrix}+\e^{-\frac2\delta\int_s^t\gamma(u)du} \frac12\begin{pmatrix}
1 &-1  \cr 
-1  & 1 \cr
\end{pmatrix}.
\end{align}
Then, Corollary \ref{cor1} together with Remark \ref{remwp} i), yield (\ref{qubitapp}) and (\ref{qubtrapro}). \qed
\end{proof}
 
\section{Generalisation}\label{secgen}

We present here a generalisation of the results concerning the perturbative regime to arbitrary high order in the adiabatic parameter. This is made possible by the use of a systematic improvement of  the adiabatic approximations of the Schr\"odinger propagator $U(t,s)$ (\ref{schr}), allowed by our general setup, see {\it e.g.} \cite{ASY, N2, JP}. We briefly present here the approach of \cite{JP} based on a hierarchy labelled by $q\in \N$ of smooth hamiltonians in $\cB(\cH)$, before spelling out the improvement it brings to the leading order results of Section \ref{weakcoup}.

\subsection{Higher order adiabatics}

Set
\bea
H^0(t)&=&H(t)\\
P_j^{0}(t)&=&P_j(t), \ \forall 1\leq j\leq d \\
K^0(t)&=&K(t).
\eea
and define the self-adjoint operator
\be
H^1(t)=H(t)-\i\eps K^0(t).
\ee
For $\eps $ small enough, the gap hypothesis {\bf Spec} holds for all $t\in [0,1]$,
and we set for  all $j\in \{1,\dots, d\}$, $\eps $ small enough, with $'$ denoting the time derivative,
\bea
P_j^1(t)&=&-\frac{1}{2\pi i}\int_{\gamma_j}(H^1(t)-z)^{-1}\, dz,  \\
K^1(t)&=&\sum_{1\leq j\leq d}{P_j^{1}}'(t)P_j^{1}(t).
\eea
Note that $H^1$, $P^1$,  and $K^1$ are $\eps$-dependent and smooth on $[0,1]$.
We define inductively, for $\eps$ small enough, the following hierarchy of operators
for $q\geq 1$, and all $j\in \{1,\dots, d\}$
\bea\label{hier}
H^q(t)&=&H(t)-\i \eps K^{q-1}(t)\\ \label{pq}
P_j^q(t)&=&-\frac{1}{2\pi i}\int_{\gamma_j}(H^q(t)-z)^{-1}\, dz, \\
K^q(t)&=&\sum_{1\leq j\leq d}{P_j^{q}}'(t)P_j^{q}(t).
\eea

It is proven in \cite{JP}, see also \cite{JP2}, that in our $C^\infty$ framework, the
following holds:
\begin{prop}\label{jpjmp}
For any $q\in \N^*$, there exists $\eps_q>0$ and $C_q, \kappa_q<\infty$,  such that for all $\eps\leq \eps_q$, $t\in [0,1]$, $j\in \{1,\dots, d\}$, 
$H^n(t), P_j^n(t), K^n(t)$ are well defined and smooth for all $0\leq n\leq q$. Moreover, these operators and all their derivatives admit an asymptotic expansion in powers of $\eps$ 
and the following estimates hold
\bea
&&\|K^{q}(t) -K^{q-1}(t)\|\leq C_q \eps^q\\
&&\|K^q(t)\|\leq \kappa_q.
\eea
\end{prop}
\begin{rem}\label{remiter} i) The hierarchy above was actually designed to reach exponential accuracy in the adiabatic approximation, in an analytic framework, in which case it provides an estimate on the behaviour in $q$ of the constant $C_q$.\\
ii) At $t=0$, the second point of assumption {\bf Reg} ensures that for all $q\in \N^*$, $H^q(0)=H(0)$ and $P^q(0)=P(0)$.\\
iii) For any time $t$, $q\in\N^*$, $1\leq j\leq d$,  $P_j^q(t)=P_j(t)+\ode(\eps)$, by perturbation theory.
\end{rem}

Let $\eps<\eps_q$ and consider the unitary propagator $V_q(t,s)_{0\leq s,t\leq 1}$, defined as the solution to
\bea\label{vq}
\left\{
\begin{matrix}
\i\eps \partial_t V_q(t,s)=(H^q(t)+\i \eps K^q(t))V_q(t,s), \\ 
 V_q(s,s)=\mathbb I, \ \ 0\leq s, t \leq 1.\hfill
\end{matrix}\right.
\eea
As is well known, see \cite{K2, Kr}, $V_q$ also satisfies
\be\label{interq}
V_q(t,s)P_j^q(s)=P_j^q(t)V_q(t,s),\ \ 0\leq s, t\leq 1.
\ee 
Note that since $H^q=H-\i\eps K^{q-1}$, we get that 
\be
H^q(t)+\i\eps K^q(t)=H(t)+\i\eps (K^q(t)-K^{q-1}(t))
\ee
is a smooth perturbation of $H(t)$. Thus, the difference between $U(t,s)$ and $V_q(t,s)$ reads
\be\label{intform}
U(t,s)-V_q(t,s)=-\int_s^t V_q(t,r) (K^q(r)-K^{q-1}(r))U(r,s) dr.
\ee
This identity and the previous proposition immediately yield 
\be\label{adq}
\|U(t,s)-V_q(t,s)\| \leq C_q|t-s|\eps^q.
\ee
We improve the error term to $\ode(\eps^{q+1})$ by performing an integration by parts on (\ref{intform}) (see the Appendix), at the cost of slightly altering the definition of $V_q$: 
Let $\hat V_q(t,s)$ be the unitary  solution to 
\bea\label{vqhat}
\left\{
\begin{matrix}
\i\eps \partial_t \hat V_q(t,s)=(H^q(t)+\i \eps (K^q(t))+\cD_qK^{q-1}(t))\hat V_q(t,s), \\ 
 \hat V_q(s,s)=\mathbb I, \ \ 0\leq s, t \leq 1,\hfill
\end{matrix}\right.
\eea
where 
\be
\cD_qK^{q-1}(t)= \sum_{1\leq j\leq d}P^q_j(t)K^{q-1}(t)P^q_j(t),
\ee
and, by convention, $K^{-1}= 0$ to recover $\hat V_0=V$.
This allows us to get the following generalisations of (\ref{adiab0}) and (\ref{inter}), see \cite{JP2}. 
\begin{prop}\label{adiabH} Under assumptions {\bf Reg} and {\bf Spec},  for all $q\in\N$, there exists $\eps_q>0$ and $c_q<\infty$ such that  for all $\eps<\eps_q$, all $j\in \{1,\dots, d\}$, and for all  $(t,s)\in [0,1]^2$
\bea\label{interqr}
&&\hat V_q(t,s)P_j^q(s)=P_j^q(t)\hat V_q(t,s)\nonumber \\
&&\|U(t,s)-\hat V_q(t,s)\| \leq c_q\eps^{q+1}
\eea
\end{prop}

\begin{rem} \label{remt}
i) As a consequence, the quantum evolution follows the instantaneous subspace $P_j^q(t)\cH$, up to an error of order $\eps^{q+1}$: 
$\|P_k^q(t)U(t,s)P_j^q(s)\|=\ode(\eps^{q+1})$ if $j\neq k$.\\
ii) The loss of  factor $|t-s|$ stems from the integration by parts procedure, see (\ref{adq}). Again, for $s=0$, we have
$\|U(t,0)-\hat V_q(t,0)\| \leq c_q t \eps^{q+1}$.\\
iii) For $s=0$, $j\neq k$, $\|P_k^q(t)U(t,0)P_j(0)\|=\ode(t\eps^{q+1})$, since $P_j^q(0)=P_j(0)$.

\end{rem}

\subsection{Higher Order Adiabatic Dyson expansion}\label{return2}

Making use of the adiabatic approximation $\hat V_q(t,s)$ of $U(t,s)$ on $\cH$  leads to an approximation of $\cU^0(t,s)$ on $\cT(\cH)$ to order $\ode(\eps^{q+1})$ and to the improvement of Proposition \ref{propweak} given in Proposition  \ref{propweaksuper}.

Let $\eps\leq \eps_q$ and define the isometric operator on $\cT(\cH)$ (and on $\cB(\cH)$)
\be\label{defv0}
\cV^0_q(t,s)(\rho)=\hat V_q(t,s) \rho \hat V_q{}^*(t,s), \ \ \rho\in \cT(\cH).
\ee
Then, for $c_q$ given in Proposition \ref{adiabH}, we get  
\bea\label{vtilteq}
\|\cU^0(t,s)-\cV^0_q(t,s)\|_\tau\leq 2c_q\eps^{q+1},
\eea
and the same holds for the subordinate operator norm  on $\cB(\cH)$.
If $\dim P_j(0)<\infty$, then $P_j^q(t)$ belongs to $\cT(\cH)$ for all $t\in[0,1]$ and $q\geq 0$, so that 
\bea\label{superad}
\cU^0(t,s)(P_j^q(s))&=&P^q_j(t)+\ode(\eps^{q+1}), \\
\cU^0(t,0)(P_j(0))&=&P^q_j(t)+\ode(t\eps^{q+1}),
\eea
see Remark ii), \ref{remt}.
If $P_j(0)$ is not trace class, the estimates above hold in operator norm.
Consequently, the first equality in equation (\ref{dyson}) and the above yields the following estimate of the propagator $(\cU(t,s))_{0\leq s\leq t\leq 1}$:
\begin{prop}\label{propweaksuper} Under  assumptions {\bf Reg} and {\bf Spec}, for any $N\geq 1$, any $q\geq 1$, there exist $\eps_q>0$, $c_q<\infty$ (given in Propositions \ref{jpjmp} and \ref{adiabH}), such that for all $0\leq s\leq t \leq 1$,
for all $\eps<\eps_q$, all $g\geq 0$, the propagator 
$\cU(t,s)\in \cB(\cT(\cH))$ satisfies  (with the convention $s_0=t$),
\begin{align}\label{expge}
 \cU(t,s)&=\cV^0_q(t,s)\nonumber \\
 &+\sum_{n=1}^N(g/\eps)^n\int_s^t\int_s^{s_1}\dots\int_s^{s_{n-1}} \cV^0_q(t,s_1)\circ \cL^1_{s_1}\circ \cV^0_q(s_1,s_2)\circ \cL^1_{s_2}\dots \circ \cL^1_{s_n}\circ \cV^0_q(s_n,s) ds_{n}\dots ds_{2}ds_{1} \nonumber\\
 &+R^q_{N+1}(t,s,\eps,g)
 \end{align}
 where, with $L_1=\sup_{0\leq s\leq 1}\| \cL_s^1\|_\tau$,
 \begin{align}\label{1stq}
\|R^q_{N+1}(t,s,\eps,g)\|_\tau\leq 2 c_q\eps^{q+1}\e^{2(t-s)L_1(1+2 c_q\eps^{q+1})g/\eps}+\frac{(L_1(t-s))^{N+1}}{(N+1)!}(g/\eps)^{N+1}.
\end{align}
In particular, if $g/\eps\leq 1$ and $\eps^{q+1}\leq 1/(2c_q)$,
 \begin{align}
\|R^q_{N+1}(t,s,\eps,g)\|_\tau&\leq  2\e^{4L_1}\left(c_q\eps^{q+1}+((t-s)g/\eps)^{N+1}\right)\nonumber\\
&=\ode_q(\eps^{q+1}+(g/\eps)^{N+1}).
\end{align}
where the notation stresses the dependence in the sole order $q$ of the constants involved.
\end{prop}
\begin{rem}\label{remu0v0}
If one keeps $\cU^0$ instead of $\cV_q^0$ in the first term of the RHS of (\ref{expge}),  the  estimate on the remainder reads
$\|R^q_{N+1}(t,s,\eps,g)\|_\tau=\ode_q((g/\eps)(\eps^{q+1}+(g/\eps)^{N}))$. 
\end{rem}
\begin{proof}
We replace $\cU^0$ by its approximation $\cV^0_q$  in each term of the Dyson series, and collect the error terms. With $\Delta_q=\|\cU^0(t,s)-\cV_q^0(t,s)\|_\tau$, the trace norm of the difference of the term of order $n\geq 1$ in (\ref{dyson}) with that of order $n$ in (\ref{expge})  is bounded above as in (\ref{basesdya}) with $\Delta_q$ in place of $\Delta$.
Then, summing over all $n\in \N$ yields the first term in (\ref{1stq}), while the second one stems from the term of order $N+1$ in (\ref{dyson}).
 The second estimate is a consequence of  $g/\eps \leq 1$, $t-s\leq 1$, $1+2 c_q\eps^{q+1}\leq 2$ and  $\frac{\alpha^{m}}{m!}\leq e^{\alpha}$, for all $m\geq 1$, $\alpha>0$.
\qed
\end{proof}

Specialising to the leading order term in $g/\eps$, and taking into account Remark \ref{remu0v0}  above, we get 
\begin{cor} Under the assumptions of Proposition \ref{propweaksuper},  for $\eps\leq \tilde \eps_q= \min(\eps_q, 1/(2c_q)^{1/(q+1)})$ and  $g/\eps\leq 1$, 
\begin{align}\label{ordeg/esuper}
 \cU(t,s)&=\cU^0(t,s)+\frac{g}{\eps}\int_s^t\cV^0_q(t,s_1)\circ \cL^1_{s_1}\circ \cV^0_q(s_1,s) ds_{1} +\ode(g\eps^{q}+g^2/\eps^{2}).
 \end{align}
\end{cor}
\medskip

Let  $0\leq \rho_j\in \cT(\cH)$ be a state such that $\rho_j=P_j(0)\rho_jP_j(0)$, and recall the definition (\ref{lind}) of the dissipator.
For any $q\in \N$, the transition probability between $P_j^q(0)\cH = P_j(0)\cH$ and $P_k^q(t)\cH$, $j\neq k$, induced by the Lindbladian dynamics (\ref{lindeq}) reads $\tr(P^q_k(t)\cU(t,0)(\rho_j))$,
 since at initial time $s=0$, one has $P_j(0)=P^q_j(0)$.  
 Using (\ref{defv0}) and Proposition \ref{adiabH}, we have
 \be\label{rhofisuper}
 \cV^0_q(s,0) (\rho_j)=\hat V_q(s,0)\rho_j\hat V_q(0,s)=P^q_j(s)\hat V_q(s,0)\rho_j\hat V_q(0,s)P^q_j(s),
 \ee
so that with $P_j^q(t)P^q_k(t)=0$,  we get from (\ref{ordeg/esuper})
 \begin{align}\label{traproq}
 \tr(P^q_k(t)\cU(t,0)&(\rho_j))= \tr(P^q_k(t)\cU^0(t,0)(\rho_j))\\ 
 &+\frac{g}{\eps}\sum_{l\in I}\int_0^t \tr (P^q_k(s)\Gamma_l(s) \hat V_q(s,0)\rho_j \hat V_q(0,s)\Gamma_l^*(s)P^q_k(s))ds +\ode(g\eps^{q}+g^2/\eps^{2}).\nonumber
 \end{align}
The Hamiltonian adiabatic transition probability between these subspaces is of order $\eps^{2(q+1)}$ according to (\ref{interqr}), whereas the  effect of the environment  is of order $g/\eps$: 
due to (\ref{vtilteq}), $ \cV^0_q(t,0) (\rho_j)= \cV^0(t,0) (\rho_j)+\ode(\eps)$ in trace norm, so that by Lemma \ref{expadiab0} the dependence in $\eps$ of $\hat V_q(s,0)\rho_j\hat V_q(0,s)$ disappears to leading order when $\sigma_j(t)=\{e_j(t)\}$.   In case $P_j(0)$ has finite rank, choosing $\rho_j=P_j(0)/\dim(P_j(0))$ yields the  integrand (see iii), Remark \ref{remiter}) :
\be
 \tr (P^q_k(s)\Gamma_l(s)P^q_j(s)\Gamma_l^*(s)P^q_k(s))=\tr (P_k(s)\Gamma_l(s)P_j(s)\Gamma_l^*(s)P_k(s))+\ode(\eps).
 \ee
Hence,  the correction term prevents the solution from following the instantaneous subspace $P^q_j(t)$ up to an error of order $\eps^{(q+1)}$, unless $g\simeq \eps^{q+2}$. 
\medskip

The coherences with respect to the iterated projectors of the integral term in (\ref{ordeg/esuper}) also vanish to leading order, thanks to (\ref{vtilteq}) and Lemma \ref{vancor2}:
\begin{lem}\label{vancor2super} Assume {\bf Reg}, {\bf Spec} and let $\rho_j=P_j(0)\rho_j P_j(0)$ be a state and suppose $\sigma_j(t)=\{ e_j(t)\}$ for all $t\in [0,1]$.
For any $1\leq n\neq m\leq d$, and  $q\geq 0$, for $\eps$ small enough
\be\label{cohvasuper}
\frac{g}{\eps}P^q_n(t)\int_0^t\cV^0_q(t,s)\circ \cL^1_{s}\circ \cV^0_q(s,0) (\rho_j)ds\, P^q_m(t)=\ode (g).
\ee
\end{lem}
Without going into the details, we note that a similar result holds for each term in (\ref{expge}). 
\medskip

Let us close this section by justifying the adiabatic expressions used throughout the paper for the populations and coherences of $\cU^0(t,0)(\rho_j)$, making use of the hierarchy (\ref{hier}).
\medskip

\noindent
{\bf Proof of Proposition \ref{propuread}:}

\noindent
Thanks to (\ref{vtilteq}) and Proposition \ref{adiabH} for $q=2$ we have 
\begin{align}
 P_k(t)\cU^0(t,0)(\rho_j)P_k(t)&=P_k(t)P_j^2(t)\hat V_2(t,0)\rho_j \hat V_2(0,t) P_j^2(t)P_k(t)+\ode(\eps^3)\nonumber\\
 &=P_k(t)P_j^2(t)\cV_2^0(t,0)(\rho_j)P_j^2(t)P_k(t)+\ode(\eps^3).
 \end{align}
For $j\neq k$, we have (dropping the variable $t$ in the notation) 
\begin{align}
P_kP_j^2&=P_k(P_j^2-P_j)=P_k(P_j^2-P_j^1+P_j^1-P_j).
\end{align}
By  perturbation theory see {\it e.g.} \cite{K2}, Proposition \ref{jpjmp} implies for $\eps$ small enough, 
\be 
P_j^q-P^{q-1}_j=\ode(H^{q}-H^{q-1})=\ode(\eps(K^{q-1}-K^{q-2}))=\ode(\eps^{q}),
\ee
so that 
\begin{align}
P_kP_j^2&=P_k(P_j^1-P_j)+\ode(\eps^2).
\end{align}
Then, $H^1=H-\i\eps K$ with $K$ given by (\ref{multik}) yields,
\be
P_j^1-P_j=-\frac{\eps}{2\pi}\oint_{\gamma_j}R(z)KR(z)dz+\ode(\eps^2)
\ee
so that making use of (\ref{vtilteq}), to write $\cV_2^0(t,0)(\rho_j)=\cV^0(t,0)(\rho_j)+\ode(\eps)$, we have
\begin{align}
P_kP_j^2\cV_2^0(\rho_j)P_j^2P_k&=-\frac{\eps^2}{(2\pi)^2}P_k \oint_{\gamma_j}R(z)P_k'R(z)dz \, \cV^0(\rho_j) \oint_{\gamma_j}R(z)P_k'R(z)dz P_k +\ode(\eps^3).
\end{align}
Hence we get (\ref{puread}) with $\cV^0(t,0)(\rho_j)$ in place of $\tilde \rho_j(t)$, as in Remark \ref{remtiltra}. Finally, assuming
$\sigma_j(t)=\{e_j(t)\}$, Lemma \ref{expadiab0} yields (\ref{puread}).
In case  $HP_k=e_kP_k$ and $HP_j=e_jP_j$, so that $R(z)P_n=P_n/(e_n-z)$ for $n\in \{j,k\}$ and $z\in \rho(H)$, a direct application of Cauchy formula yield (\ref{pureadev}). 

The expressions (\ref{cohmn}) for the coherences are proven quite similarly. 
\qed

 \section{Appendix: Integration by Parts}\label{ipp}
 
 We present here a reformulation of the integration by parts argument used in \cite{ASY} to prove the adiabatic theorem of quantum mechanics,  suited to our setup.
 
Let $\cZ$ be a Banach space and assume $\cG:[0,1]\ra \cB(\cZ)$, $\cK:[0,1]\ra \cB(\cZ)$ are bounded operator valued $C^\infty$ functions on $[0,1]$, in the norm sense.
Let $\eps >0$, and consider the two-parameter propagators $(\cX(t,s))_{1\leq s\leq t\leq 1}$ and $(\cY(t,s))_{1\leq s\leq t\leq 1}$, solution to the equations
 \begin{align}\label{x}
\left\{\begin{matrix}
\eps \partial_t\cX(t,s)=\cG(t)\cX(t,s), \hfill\cr
\cX(s,s)=\un, \ \ 0\leq s \leq t \leq 1,\hfill 
\end{matrix} \right.
\end{align}
and 
\begin{align}\label{y}
\left\{\begin{matrix}
\eps \partial_t\cY(t,s)=(\cG(t)+\eps \cK(t))\cY(t,s), \cr
\cY(s,s)=\un, \ \ 0\leq s \leq t \leq 1.\hfill 
\end{matrix} \right.
\end{align}
The smooth propagators $\cX(t,s)$ and $\cY(t,s)$ are determined by the corresponding Dyson series, both depend on $\eps>0$ with norms that diverge as $\eps\ra 0$, {\it a priori }.
Moreover, they satisfy the integral relation
\be\label{intppadiab}
\cX(t,r)=\cY(t,r)-\int_r^t \cY(t,s) \cK(s) \cX(s,r)ds, \ \ \forall 1\geq t\geq r \geq 0.
\ee

Assume the existence of gaps in the spectrum of $\cG(t)$, uniformly in $t\in[0,1]$. For $d\in \N^*$,
\be
\sigma(\cG(t))=\cup_{1\leq j\leq d}\, \sigma_j(t)\subset \C,  \ \  \inf_{t\in[0,1], 1\leq j\neq k\leq d}{\rm dist} (\sigma_j(t),\sigma_k(t))\geq G >0.
\ee  
Consider the corresponding spectral projector 
\be\label{gari}
\cP_j(t)=-\frac{1}{2 \pi \i}\int_{\gamma_j} (\cG(s)-z)^{-1}dz,
\ee
where $\gamma_j$ is a simple loop in $\rho(\cG(t))$, the resolvent set of $\cG(t)$, encircling $\sigma_j(t)$ and such that for all $k\neq j$, $\mbox{int}\gamma_j\cap \sigma_k(t)=\emptyset$. 
For $\cB:[0,1]\ra \cB(\cZ)$, a smooth bounded operator valued function, define for any $t\in[0,1]$
\be\label{oprb}
\cR_j(\cB)(t)=-\frac{1}{2\pi\i } \oint_{\gamma_j} (\cG(t)-z)^{-1}\cB(t)(\cG(t)-z)^{-1}dz,
\ee
with the same loop $\gamma_j$ as in (\ref{gari}). This operator is smooth as well, and satisfies the identity
\be\label{idgr}
[\cG(t), \cR_j(\cB)(t)]=[\cB(t),\cP_j(t)].
\ee
\begin{rem}\label{rem81} If $\sigma_k(t)=\{g_k(t)\}$ for all $1\leq k\leq d$, $(\cG(t)-z)^{-1}=\sum_{1\leq k \leq d} \cP_k(t)/(g_k(t)-z)$, and 
\be\label{rbdis}
\cR_j(\cB)(t)=\sum_{1\leq k\leq d\atop k\neq j} \frac{\cP_j(t)\cB(t)\cP_k(t)+\cP_k(t)\cB(t)\cP_j(t)}{g_k(t)-g_j(t)}.
\ee
\end{rem}

\begin{lem}\label{idint}
Suppose $\cK(t)$ is off-diagonal for all $t\in [0,1]$, {\it i.e.} s.t. $\cP_j(t)\cK(t)\cP_j(t)\equiv 0$, $\forall 1\leq j\leq d$. Then
\begin{align}\label{formulibp}
\cX(t,r)&-\cY(t,r)= \frac12\sum_{1\leq j\leq d} \eps \Big(\cR_j([\cK, \cP_j])(t)\cX(t,r)- \cY(t,r) \cR_j([\cK, \cP_j])(r)\Big)\\
&+\frac12\sum_{1\leq j\leq d} \eps  \int_r^t  \Big\{  \cY(t,s)\cK(s)[\cG(s), \cR_j([\cK, \cP_j])(s)] \cX(s,r) -\cY(t,s) (\partial_s\cR_j([\cK, \cP_j])(s))\cX(s,r) \nonumber
\Big\} ds.
\end{align}
\end{lem}
\begin{proof}
The operator $\cK$ being off-diagonal and  (\ref{idgr}) give 
\be
\cK(t)=\frac12\sum_{1\leq j\leq d}\big[[\cK(t), \cP_j(t)\big],\cP_j(t)]=\frac12\sum_{1\leq j\leq d}\big[\cG(t), \cR_j([\cK, \cP_j])(t)\big].
\ee
Hence, using (\ref{intppadiab}), (\ref{x}) and (\ref{y}),
\begin{align}
\cX(t,r)-\cY(t,r)
=&-\frac12\sum_{1\leq j\leq d}\int_r^t \cY(t,s) \big[\cG(s), \cR_j([\cK, \cP_j])(s)\big] \cX(s,r)ds
\end{align}
where, for each integral in the summand
\begin{align}
-\int_r^t \cY(t,s) [&\cG(s), \cR_j([\cK, \cP_j])(s)] \cX(s,r)ds\nonumber\\
=&\eps \int_r^t  \Big\{(\partial_sY(t,s)) \cR_j([\cK, \cP_j])(s)\cX(s,r) + \cY(t,s)\cK(s)\big[\cG(s), \cR_j([\cK, \cP_j])(s)\big] \cX(s,r) \nonumber\\ 
&+ \cY(t,s) \cR_j([\cK, \cP_j])(s)\partial_s\cX(s,r) \Big\} ds.
\end{align}
Thanks to the smoothness of all operators in the integrand,  we have
\begin{align}
(\partial_s\cY(t,s) )&\cR_j([\cK, \cP_j])(s)\cX(s,r)+\cY(t,s) \cR_j([\cK, \cP_j])(s)\partial_s\cX(s,r)\nonumber\\
&=
\partial_s(\cY(t,s) \cR_j([\cK, \cP_j])(s)\cX(s,r))
-\cY(t,s) (\partial_s\cR_j([\cK, \cP_j])(s))\cX(s,r),
\end{align}
which yields the sought for identity. \qed
\end{proof}

As a corollary of Lemma (\ref{idint}), if one of the propagators $(\cX(t,s))_{0\leq s\leq t\leq 1}$ or $(\cY(t,s))_{0\leq s\leq t\leq 1}$ is uniformly bounded in $\eps$, so is the other, and their difference goes to zero with $\eps$:

\begin{cor}\label{bddsmall}
Assume  $\exists\ \eps_1>0,  C_1<\infty$ such that $\sup_{0<\eps\leq \eps_1\atop 0\leq s\leq t\leq 1}\|\cX(t,s)\|\leq C_1$. Then, $\exists \ \eps_2>0,  C_2<\infty$ such that $\sup_{0<\eps\leq \eps_2 \atop 0\leq s\leq t\leq 1}\|\cY(t,s)\|\leq C_2$. The same statement holds for $\cX$ and $\cY$ exchanged.
Moreover,  $\exists\  C_3<\infty$ such that for all $ \eps <\eps_2$, 
\be\label{difgoze}
\sup_{0\leq s\leq t\leq 1}\|\cX(t,s)-\cY(t,s)\|\leq C_3 \eps,
\ee
and, whenever $\cK(0)=0$, there exists $C_4<\infty$ so that for all $t\in[0,1]$
\be\label{laboun}
\|\cX(t,0)-\cY(t,0)\|\leq C_4 t \eps.
\ee
\end{cor}
\begin{rem}
If both $\cX(t,s)$ and $\cY(t,s)$ are a priori uniformly bounded, estimate (\ref{difgoze}) holds for all $\eps$.
\end{rem}
\begin{proof}
Set 
\begin{align}
2C_0=\max\Big(& \sum_{1\leq j\leq d}\sup_{0\leq s\leq 1}\|\cR_j([\cK, \cP_j])(s)\|, \sum_{1\leq j\leq d}\sup_{0\leq s\leq 1}\|\partial_s \cR_j([\cK, \cP_j])(s)\|,
 \nonumber\\
& \sum_{1\leq j\leq d}\sup_{0\leq s\leq 1}\|\cK(s)\big[\cG(s), \cR_j([\cK, \cP_j])(s)\big]\|\Big)
\end{align}
and consider $\eps<\eps_1$. Lemma (\ref{idint}) yields the bound
\begin{align}
\|\cY(t,r)\|&\leq C_1+\eps C_0 (\|\cY(t,r)\|+C_1)+\eps 2C_0C_1\sup_{0\leq s \leq t\leq 1} \|\cY(t,s)\|
\end{align}
so that, taking the supremum over $0\leq r\leq t\leq 1$ and for $\eps < \min (\eps_1, 1/(2C_0(1+2C_1))) := \eps_2$, we get in turn
%
%
\begin{align}
\sup_{0\leq r \leq t\leq 1}\|\cY(t,r)\|&\leq \frac{1}{(1-\eps C_0(1+2C_1))}(C_1(1+\eps C_0))
\leq C_1\frac{3+4C_1}{1+2C_1}\equiv C_2.
\end{align}
Then, inserting this estimate into (\ref{formulibp}), one gets, uniformly in $0\leq s\leq t\leq 1$, 
\be
\|\cX(t,s)-\cY(t,s)\|\leq C_3 \eps,
\ee
with $C_3=C_0(C_1+C_2+2C_1C_2)$.
Finally, for the initial time $s=0$, the integrated contribution in (\ref{formulibp}) reduces to $\frac12\sum_{1\leq j\leq d} \eps \cR_j([\cK, \cP_j])(t)\cX(t,0)$ when $\cK(0)=0$,
and since either $\cR_j([\cK, \cP_j])(t)=0$ or
\be
\cR_j([\cK, \cP_j])(t)=\int_0^t\partial_s \cR_j([\cK, \cP_j])(s)ds,
\ee 
we have in any case
\be
\Big\| \frac12\sum_{1\leq j\leq d}\cR_j([\cK, \cP_j])(t)\Big\|\leq tC_0.
\ee
The integral term in (\ref{formulibp}) is of order  $\eps t$, so that the  bound (\ref{laboun}) holds with $C_4=C_0C_1(1+2C_2)$.

The fact that $\cX$ and $\cY$ can be exchanged in all arguments above follows from  the structure of the RHS of (\ref{formulibp}).
\qed
\end{proof}

\medskip

\noindent
{\bf Proof  of  Lemma \ref{atqm}:}\\
We briefly prove estimate (\ref{adiab0}). 
Here the Banach space is $\cZ=\cH$, the propagators are $\cX=U$, $\cY=V$, and the generators are constructed with $\cG=-iH$ and $\cK=K=\sum_{1\leq j\leq d} P'_jP_j$.

Using $P(t)P'(t)P(t)\equiv 0$ for any smooth projector $P(t)$, one has the identities $P_j(t)K(t)P_j(t)\equiv 0$, for all $1\leq j\leq d$, and actually both $U$ and $V$ are bounded {\it a priori}, since they are unitary. Moreover, $K(0)=0$, under {\bf Reg}. Hence Lemma \ref{atqm} derives from Corollary \ref{bddsmall}. \qed

\medskip

\noindent
{\bf Proof  of  Lemma \ref{vancor2}:}\\
As a second application, we derive here estimate  (\ref{cohva}). 
We need to show that 
\be
P_n(t)\int_0^t\cV^0(t,s)\circ \cL^1_{s}\circ \cV^0(s,0) (\rho_j)ds\, P_m(t)=\ode(\eps).
\ee
We first note that  by Lemma \ref{expadiab0}, 
$\cV^0(s,0) (\rho_j)=\tilde \rho_j(s)$, where $\partial_s \tilde \rho_j(s)= [K(s), \tilde \rho_j(s)]$ is continuous in trace norm and $\eps$-independent. Moreover, 
\be\label{spliint}
P_n(t)\cV^0(t,s)(\cdot)P_m(t)=P_n(t)\big(\cV^0(t,0)\circ \cV^0(0,s)(P_n(s) \cdot P_m(s))\big)P_m(t),
\ee
thanks to the intertwining property (\ref{interqr}).
Using the definition of $\cL_s^1$ , we have
\begin{align}
P_n(s) \cL^1_{s}\circ \cV^0(s,0) (\rho_j) P_m(s)=&\sum_{l\in I}P_n(s)(\Gamma_l(s) \tilde \rho_j(s)\Gamma_l^*(s)P_m(s) \\
&-\delta_{mj}\frac12P_n(s)\Gamma_l^*(s)\Gamma_l(s) \tilde \rho_j(s)P_m(s) - \frac12 \delta_{nj}P_n(s) \tilde \rho_j(s)\Gamma_l^*(s)\Gamma_l(s) P_m(s),\nonumber
\end{align}
where all terms are independent of $\eps$. Hence, to get the result, we are lead to show that for a smooth trace class operator $[0,1]\ni s\ra F(s)$, such that $\partial_sF(s)\in \cT(\cH)$, independent of $\eps$, and $n\neq m$,
\be
\int_0^tV^0(0,s) P_n(s)F(s)P_m(s)V^0(s,0)ds\, =\ode(\eps).
\ee
We have thanks to (\ref{idgr}) with $\cG=H$
\be
P_n(s)F(s)P_m(s)=[P_n(s)F(s)P_m(s), P_m(s)]=[H(s),\cR_m(P_nFP_m)(s)]
\ee
so that, by a slight variation of Lemma \ref{idint}
\begin{align}
\int_0^tV^0&(0,s) [H(s),\cR_m(P_nFP_m)(s)]V^0(s,0)ds=-\i\eps V^0(0,s) \cR_m(P_nFP_m)(s)]V^0(s,0)|_0^t\\ \nonumber
&+\i\eps \int_0^tV^0(0,s) \Big\{ \partial_s\cR_m(P_nFP_m)(s)-K(s)\cR_m(P_nFP_m)(s)+\cR_m(P_nFP_m)(s)K(s)\Big\}V^0(s,0)ds.
\end{align}
As $F(s)$ and its derivative are trace class, the expression above is $\ode(\eps)$ in trace norm. \qed\\
Let us finally note that, making use of the projectors appearing (\ref{spliint}), we can further integrate by parts the last integral term, provided $\partial_sF(s)$ is continuously differentiable in trace norm,  in which case
\begin{align}
\int_0^tV^0(0,s) P_n(s)F(s)P_m(s)V^0(s,0)ds
&= \i\eps \big(\cR_m(P_nFP_m)(0)-V^0(0,t) \cR_m(P_nFP_m)(t)V^0(t,0)\big)+\ode(\eps^2).
\end{align}

\noindent
{\bf Proof of Lemma \ref{gendiag}:}

\noindent
For any  $\rho\in \cT(\cH)$, we need to consider 
\begin{align}\label{zeroth}
\int_0^t\int_0^{s_1}\dots\int_0^{s_{n-1}} \cV^0(t,s_1)\circ \cL^1_{s_1} \circ \cV^0(s_1,s_2) \circ \cL^1_{s_2}\dots  \cL^1_{s_n} \circ \cV^0(s_n,0) (\rho)ds_{n}\dots ds_{2}ds_{1}.
\end{align}
Noting with (\ref{defv00}) that for any $0\leq s,t\leq 1$ 
$
\cV^0(t,s)(\rho)=\cV^0(t,0)\circ \cV^0(0,s)(\rho),
$
we can write (\ref{zeroth}) as $\cV^0(t,0)(I_n(t))$, where $I_n(t)\in \cT(\cH)$ is defined inductively by
\begin{align}\label{defin}
I_n(t)&=\int_0^t  \cV^0(0,s_1)\circ \cL^1_{s_1} \circ\cV^0(s_1,0)(I_{n-1}(s_1))ds_1,\nonumber\\
I_1(t)&=\int_0^t  \cV^0(0,s_n) \circ\cL^1_{s_n} \circ\cV^0(s_n,0)(\rho)ds_n.
\end{align}
Lemma \ref{vancor2} shows the existence of $C_1<\infty$, such that for all $0\leq t\leq 1$ and $\eps$ small enough, 
\be
\| I_1(t)-\cP_0(0)(I_1(t))\|_1\leq \eps C_1 \|\rho \|_1.
\ee
Let us show by induction that for for each $n$, there exists $C_n<\infty$ so that for $\eps$ small enough
\be
\| I_n(t)-\cP_0(0)(I_n(t))\|_1\leq \eps C_n \|\rho \|_1.
\ee
Assuming the result for $n\geq 1$, we consider the step $n+1$. We  get
\begin{align}\label{714}
 I_{n+1}(t)-\cP_0(0)&(I_{n+1}(t))=\sum_{1\leq j\neq k\leq d}P_j(0)I_{n+1}(t)P_k(0)\\
 &=\sum_{1\leq j\neq k\leq d}\int_0^t P_j(0)\big\{\cV^0(0,s) \circ \cL^1_{s} \circ \cV^0(s,0)(I_n(s))\big\}P_k(0)ds\nonumber \\
 &= \sum_{1\leq j\neq k\leq d}\int_0^t P_j(0)\big\{\cV^0(0,s) \circ \cL^1_{s} \circ \cV^0(s,0)\circ \cP_0(0)(I_n(s)))\big\}P_k(0)ds+\ode_n(\|\rho\|_1 \eps), \nonumber
\end{align}
by the induction hypothesis and recalling the operator $\cV^0(t,s)$ is isometric and $\cL_s^1$ is uniformly bounded. 
Then we observe that for any $j\neq k$, and any $[0,1]\ni s\mapsto A(s)\in\cT(\cH)$,  $C^1$ in trace norm,  see Lemmas \ref {decvwp}, \ref{lemw} and (\ref{vpw})
\begin{align}\label{estoff}
\int_0^t P_j(0)\big\{\cV^0(0,s)& (A(s))\big\}P_k(0)ds=\int_0^t e^{\frac{\i}{\eps}\int_{0}^s(e_j-e_k)(u)du}P_j(0)\big\{\cW^0(0,s) (A(s))\big\}P_k(0)ds\nonumber\\
&=\frac{-\i\eps}{e_j(s)-e_k(s)}e^{\frac{\i}{\eps}\int_{0}^s(e_j-e_k)(u)du}P_j(0)\big\{\cW^0(0,s) (A(s))\big\}P_k(0)\Big|_0^t\nonumber\\
&+\int_0^t\i\eps e^{\frac{\i}{\eps}\int_{0}^s(e_j-e_k)(u)du}P_j(0)\partial_s\Big(\frac{\cW^0(0,s) (A(s))}{e_j(s)-e_k(s)}\Big)P_k(0)ds.
\end{align}
The trace norm of the RHS is bounded above by $\eps c (\sup_{0\leq s\leq 1}\|A(s)\|_1+\sup_{0\leq s\leq 1}\|\partial_s A(s)\|_1)$, where $c$ is a constant independent of $\eps$.
The integral term of the RHS of (\ref{714}) has the form (\ref{estoff}) with
\be
A(s)= \cL^1_{s} \circ \cV^0(s,0)\circ \cP_0(0)(I_n(s))= \cL^1_{s} \circ \cW^0(s,0)\circ \cP_0(0)(I_n(s)),
\ee
where $\cW^0(s,0)$ and $\cL_s^1$ are smooth, independent of $\eps$ and bounded on $\cT(\cH)$, while $I_n(s)$ and $\partial_t I_n(s)$ are continuous and bounded in trace norm by a constant (uniform in $\eps$) time $\|\rho\|_1$, see (\ref{defin}), which ends the proof. \qed

 \end{document}